\documentclass[envcountsame]{llncs}
\usepackage{amssymb,amsmath,latexsym,amsfonts}
\usepackage{graphicx,subcaption}
\usepackage{url}

\newcommand{\fuera}[1]{}

\usepackage{calc}
\usepackage[shortlabels]{enumitem}
\usepackage{tikz}
\usepackage{graphicx,adjustbox}
\usepackage{xcolor}
\usepackage[noend]{algpseudocode}
\usepackage{url}
\usetikzlibrary{shapes,snakes}
\usetikzlibrary{trees}
\usetikzlibrary{decorations.pathmorphing}
\usetikzlibrary{decorations.markings}
\usetikzlibrary{arrows}
\usetikzlibrary{positioning}
\usetikzlibrary{shapes.geometric}
\usetikzlibrary{decorations.pathreplacing}
\usepackage{tabularx}
\newcolumntype{Y}{>{\centering\arraybackslash}X}

\definecolor{light-gray}{gray}{0.75}
\SetLabelAlign{parright}{\parbox[t]{\labelwidth}{\raggedleft#1}}

\DeclareMathOperator*{\argmin}{arg\,min}

\begin{document}

\title{Group Centrality for Semantic Networks: \\
a SWOT analysis featuring Random Walks
%%based on Random Walks and Classical Measures
}
%\titlerunning{Graph Summaries}  % abbreviated title (for running head)
%                                     also used for the TOC unless
%                                     \toctitle is used
%
\author{Camilo Garrido \and Ricardo Mora \and Claudio Gutierrez}
\institute{
Department of Computer Science, Universidad de Chile, Chile\\
\email{\{cgarrido,rmora,cgutierr\}@dcc.uchile.cl}
}

\maketitle              % typeset the title of the contribution

\begin{abstract}
Group centrality is an extension of the classical notion of
centrality for individuals, to make it applicable to sets of them. We
perform a SWOT (strengths, weaknesses, opportunities and threats) analysis of the use of group centrality in semantic
networks, for different centrality notions: degree,  closeness, betweenness, giving
prominence to random walks. Among our main results stand out the
relevance and NP-hardness of the problem of finding the most central set in a semantic network
for an specific centrality measure.
\end{abstract}

%%%%%%%%%%%%%%%%%%%%%%%%%%%%%%%%%%%%%%%%%%%%%%%%%%%%%%%%%%%%%
\section{Introduction}

%%%
% \paragraph{The problem and its relevance.}

 The notion of centrality in graphs has been thoroughly studied since
it was introduced for social networks, being
also applied in other contexts ranging from physics to information
retrieval \cite{Newman}%.The concept of centrality has a wide range of applications
, but as
Freeman writes for social networks \cite{Freeman},
``there is certainly no unanimity on exactly what
centrality is or on its conceptual foundations, and there is very
little agreement on the proper procedure for its measurement.''
Over the years, a great variety of measures of centrality have been
proposed, most of them pointing to the problem of finding the most
``central'' (regarding the notion and/or the measure) node in a
network \cite{goodcita}.

 Most of classic work on centrality studies what we could call
``individual'' centrality: what is the most central node in a
network for a certain measure. This idea can be extended to a
group: Given an integer $k> 1$,
what is the most central set $S$ of nodes, of size $k$?
This is known as {\em group centrality} and was introduced
by Everett and Borgatti~\cite{everett} for social networks
in order to make individual measures of centrality work with groups of elements.

  In this paper we study the application of the idea of group centrality
to semantic networks. To the best of our knowledge, this link has not
been explored before. It can help to shed light on the idea
of a ``most significative'' group of concepts in a semantic network, which
in the age of big (semantic) data, seems to be indispensable to make sense
of notions like summary, table of contents,
compendium, etc. on huge semantic network datasets.

Our goal in this paper is to show the robustness of this application, develop its properties, and present its opportunities,
scope and possibilities and limitations.

 We perform the study based on some classical  centrality
measures like degree, closeness and betweenness \cite{Freeman}.\footnote{
Degree centrality selects the node with higher degree; closeness
indicates the node that minimizes the distance from all other nodes in the
network to it; Betweenness selects the node over which pass the higher number
of shortest paths between each pair of nodes. We will define them
formally in Section 3.
}
We present in more detail random walk centrality which is a measure
 based on the notion
of hitting time (informally: the expected length of a random walk).
This idea has been used with success in
 modeling semantic relatedness. The intuition is that the ``jump'' from
one concept to another is better modeled through random walks rather
 than using deterministic movements.  This approach has been explored in
human memory \cite{HumanMemory}, lexical relations \cite{LexicalSemantic}
and Wikipedia corpus \cite{WikiWalk}.

%%%%%%%%%%%%%%%%%%%%%%%%%%%%%%%%%%%%%%%%%%%%%%%%%%%%%%%
%%%\begin{figure}[t!]
%  \centering
%  \includegraphics[width=10cm]{proximity1c.png}
%\includegraphics[width=6cm]{proximity3_label.png}  \\
%\begin{tabular}{crcl} \hline
%$k=1$ & livinthing & -- & livingthing, \\
%$k=2$ & red, livingthing & -- & livingthing,  plant \\
%$k=3$ & flower, animal, livingthing  & --  & animal, plant, bird \\
%$k=4$ & plant, red, animal, livingthing& -- & animal, plant, bird, red  \\\hline
% \end{tabular}
%  \caption{The nework of notions; strong relations in bold plus weak
%in dashed lines.   }
%\end{figure}
%%%%%%%%%%%%%%%%%%%%%%%%%%%%%%%%%%%%%%%%%%%%%%%%%%%%%%%%%%

\begin{figure}[!h]
    \centering
    \begin{subfigure}[b]{0.43\linewidth}        %% or \columnwidth

%\begin{tikzpicture}[scale=0.75, every node/.style={scale=0.75}, line width=0.25pt]
%\begin{tikzpicture}[thick,scale=0.6, every node/.style={transform shape}]
\resizebox{5.5cm}{!}{%
\begin{tikzpicture}[scale=0.9, every node/.style={scale=1.25}]
\usetikzlibrary{arrows}
\usetikzlibrary{shapes}
\tikzstyle{every node}=[draw=black, ellipse,align=center]
\node (antlers) at (-3,3.5){antlers};
\node (deer) at (-0.6,3){deer};
\node (hooves) at (2.5,4){hooves};
\node (mammal) at (-4,2.5){mammal};
\node (dog) at (-2,1.5){dog};
\node (hair) at (-4.25,0.5){hair};
\node (animal) at (-0.5,1.8){animal};
\node (blood) at (-4,-1){blood};
\node (red) at (-3.5,-2){red};
\node (rose) at (-0.5,-2){rose};
\node (color) at (-4,-4.5){color};
\node (green) at (-0.8,-5.5){green};
\node (flower) at (-1,-3){flower};
\node (daisy) at (-1.2,-4.5){daisy};
\node (plant) at (2.5,-3){plant};
\node (tree) at (3.4,-0.1){tree};
\node (cottonwood) at (4.5,-4.2){cottonwood};
\node (leaves) at (4.7,0.5){leaves};
\node (robin) at (-1.5,-1){robin};
\node (bird) at (0,-0.5){bird};
\node (bat) at (-0.8,0.75){bat};
\node (livingthing) at (2,2){livingthing};
\node (frog) at (4.25,4){frog};
\node (chicken) at (1.5,0.5){chicken};
\node (feathers) at (1.5,-2){feathers};
\foreach \from/\to in {antlers/deer,hooves/deer,animal/deer,mammal/animal,dog/animal,hair/dog,animal/livingthing,blood/animal,blood/red,robin/red,rose/red,color/red,
											 color/green,flower/rose,daisy/flower,flower/plant,plant/tree,cottonwood/tree,leaves/tree,robin/bird,bird/feathers,bird/bat,bat/livingthing,frog/livingthing,
											 chicken/livingthing,feathers/chicken,green/plant,livingthing/tree}
 \draw (\from) -- (\to);
\end{tikzpicture}
}
				%\caption{Path of length $n$.}
				\label{pathh}
    \end{subfigure}
    \begin{subfigure}[b]{0.45\linewidth}        %% or \columnwidth
        \centering
%\begin{tikzpicture}[scale=0.5, every node/.style={scale=0.5}, line width=0.25pt]
\resizebox{6.5cm}{!}{%
\begin{tikzpicture}[scale=0.9, every node/.style={scale=1.25}]
\usetikzlibrary{arrows}
\usetikzlibrary{shapes}
\tikzstyle{every node}=[draw=black, ellipse,align=center]
\node (antlers) at (-4,3){antlers};
\node (deer) at (-2,3){deer};
\node (hooves) at (-3,4){hooves};
\node (mammal) at (0,2){mammal};
\node (dog) at (1,3){dog};
\node (hair) at (0,4){hair};
\node (animal) at (2,4){animal};
\node (blood) at (2,1){blood};
\node (red) at (3,0){red};
\node (rose) at (-3.2,-4){rose};
\node (color) at (3,-1){color};
\node (green) at (1.5,-2.5){green};
\node (flower) at (-3,-3){flower};
\node (daisy) at (0,-4){daisy};
\node (plant) at (-1,-2){plant};
\node (tree) at (-2,0){tree};
\node (cottonwood) at (-4.5,0){cottonwood};
\node (leaves) at (-2,-1){leaves};
\node (robin) at (6,0){robin};
\node (bird) at (4,2){bird};
\node (bat) at (1,0){bat};
\node (livingthing) at (-1,1){livingthing};
\node (frog) at (4.25,4){frog};
\node (chicken) at (3,3){chicken};
\node (feathers) at (5,1){feathers};
\foreach \from/\to in {robin/feathers,feathers/bird,bird/chicken,chicken/animal,animal/frog,animal/dog,dog/mammal,hair/mammal,bat/mammal,blood/mammal,blood/red,red/color,green/color,
											 plant/daisy,daisy/flower,flower/rose,plant/leaves,leaves/tree,tree/cottonwood,tree/livingthing,livingthing/mammal,mammal/deer,livingthing/deer,deer/hooves,deer/antlers,
											 flower/plant,plant/green}
\draw (\from) -- (\to);
\end{tikzpicture}
}
				%\caption{$n$-Clique.}
				\label{axiom1}
    \end{subfigure}
		\begin{tabular}{crcl} \hline
$k=1$ & $\{$ livingthing $\}$ & -- & $\{$ mammal $\}$ \\
$k=2$ & $\{$ red, livingthing  $\}$& -- &   $\{$mammal, plant $\}$ \\
$k=3$ & $\{$ bird, plant, animal $\}$ & --  & $\{$ mammal, plant, bird $\}$ \\
$k=4$ &  $\{$ plant, red, animal, livingthing $\}$ & -- & $\{$ mammal, plant, bird, tree $\}$ \\\hline
 \end{tabular}
  \caption{Networks built by novices (left)  and experts (right) over
    the same set of 25 concepts. In the table, for each $k$ it is shown the most
    central set of $k$ concepts of the novice (left) and the expert  (right) networks.}
    \label{fig:pathclique}
\end{figure}
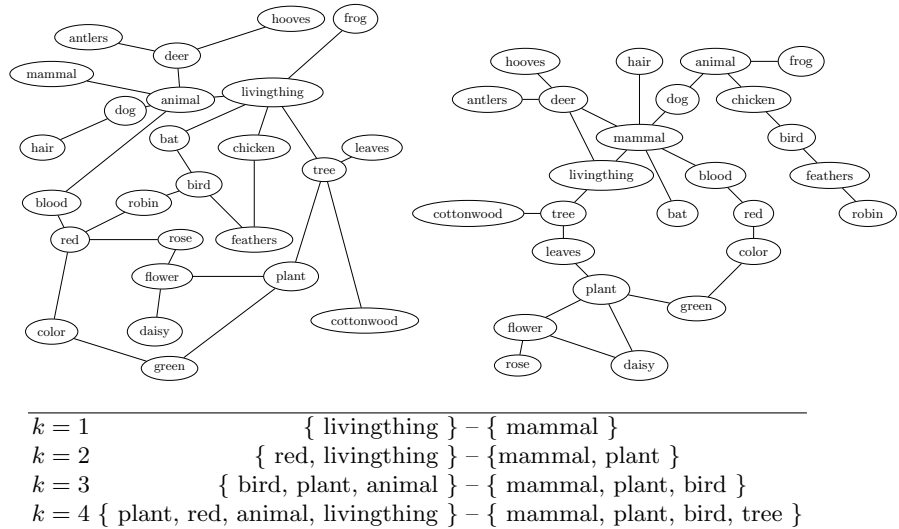
\medskip

   Let us consider the examples shown in Fig. 1. Schvaneveldt and
social scientists~\cite{proximity} made a study of the process of
construction of semantic networks. They chose 25 natural concepts and
devised a method to allow communities to establish semantic
relationships among them.
The networks were built using two different communities:
on the left by novices (psychology students) and on the right by
domain experts (biologists).
 What are the most central concepts here?
First, let us point that
the four centrality measures mentioned previously, show remarkable similarity in their outputs
(these measures do not consider labels, only the structure of the graph).
On the novices' network, the four measures indicate ``livingthing''.
As for the biologists' network they indicate ``mammal''.\footnote{
In the experiments, the biologists were asked separately to tell the
most representative concept in the list, and they chose ``mammal''.
}
%\footnote{
%These two cases were chosen by the researchers of the study. One
%could devise a continuous set of edges by assigning weights. We
%do not include in this paper weights because we did not find
%data with realiable weights assigned to edges. But all our theoretical
%results extend in a straightforward manner to weighted edges.
%}
Also, the four of them behaved remarkable similar when searching for central groups of different sizes, as shown in Fig. 1.

  What do these examples tell? First of all, they give insights on the
possible relevance and use that group centrality could have for semantic
networks. The contribution of this paper is a study of the assumptions
that lie under this optimistically/naive hypothesis, that is, an
analysis of the strengths and opportunities and weaknesses and obstacles
(known as SWOT analysis) of
applying this idea to semantic networks. (In the related work section we
will go over the work that has been done on ranking semantic networks,
that touches some facets of this problem).

The concrete contributions of this paper are:
\begin{enumerate}
\item Complexity analysis.
We show that for the four previously mentioned centrality measures, the problem of
finding a set with optimal centrality score is NP-hard (and its associated decision problem
is NP-complete).

\item  We present a small-scale study of samples of semantic networks
 --because of the hardness of the computations and the
difficulties to get a human evaluation of what
is ``central'' in a huge network, a challenge by itself--
to  show tendencies and illustrate the potential and robustness of
the idea of group centrality.

\item We develop the notion of random-walk centrality, and prove
several theoretical properties for it. In particular, we show that testing
random-walk group centrality can be done in polynomial time.
\end{enumerate}

The paper is organized as follows. In Section 2 we present the
main arguments of the SWOT analysis and related work on the subject.
In Section 3 we show the study of samples of some real semantic networks. In section 4 we include the
proofs and support for the theoretical arguments and in section 5 we present the conclusions.

\section{Group Centrality in Social Networks}

Let us begin by stating formally the notions to be studied.

%-----------------------------------------------------------------------------------------------------------------------------------------------------------

\paragraph{Basic graph notions.}

An undirected and simple graph is a pair $G=(V,E)$ where $E \subseteq [V]^2$, and $[V]^2$ is the set of all 2-elements subsets from $V$. The elements of $V$ are the \emph{vertices} of $G$, and the ones from $E$ are its \emph{edges}. When necessary, we will use the notation $V(G)$ and $E(G)$ for those sets. From now on, an element $\{u,v\} \in E$ will be denoted simply by $uv$.

% An important family of graphs are cliques:  for $n \geq 1$ a $n$-$clique$ is a graph $K_n:=(V,E)$ with $|V|=n$, such that $E=[V]^2.$

A vertex $u$ is said to be a \emph{neighbor} of another vertex $v$, when $uv \in E$. Note that the definition of $E$ implies that $v$ is also a neighbor of $u$. The set of neighbors of $v$ will be denoted by $N_G(v)$. The \emph{degree} of $v$, $d_G(v)$ is the size of $N_G(v)$.

%Let $G=(V,E)$ and $G^\prime=(V^\prime,E^\prime)$ be two graphs such that $V \subseteq V^\prime$ and $E \subseteq E^\prime$, then $G^\prime$ is said to be a \emph{subgraph} of $G$ (it is also said that $G$ \emph{contains} $G^\prime$). %For a subset $S \subseteq V$, $G[S]:=(S,\{uv \in E: u,v \in S\})$ and $G-S:=G[V \setminus S]$. Analogously, for $F \subseteq E$, $G-F:=(V,E \setminus F)$.

A graph $P_n=\big(\{v_0,v_1,...,v_n\},\{v_0v_1,v_1v_2,...,v_{n-1}v_n\}\big)$ with $n \geq 0$, where all $v_i$ are distinct is called a \emph{path}, and the number of edges in it is its \emph{length}. A {\em cycle} is a special type of path such that $v_0=v_n$. %We will call a cycle of length $n$ a $n \textendash cycle$.
A path $P_n$ in $G$, with $n \geq 1$ such that $v_0=u$ and $v_n=v$, is called a $u$-$v\ path$. Also $G$ is said to be \emph{connected} if for all distinct $u,v \in V$ a $u$-$v$ path exists in $G$. A \emph{connected component} of $G$ is a maximally connected subgraph $H$.

The \emph{distance} between $u$ and $v$ \big(denoted by $d_G(u,v)$ or just $d(u,v)$\big) is the length of the shortest $u$-$v$ path in $G$. For $S \subseteq V$, $d(u,S):=\min_{v \in S}d(u,v)$.

In what follows, all our graphs will considered to be simple, undirected and connected.

%A \emph{Tree} $T=(V,E)$ is a connected graph that does not contain any cycle. In a tree there exists one and only one $u$-$v$ path between every distinct pair of vertices $u,v$. %(which will be denoted by $uTv$)
%------------------------------------------------------------------------------------------------------------------------------------------------------------------

\paragraph{Group Centrality}

%---------------------------------------------------------------------------------------------------------------------------------------------------------------

%Let $G=(V,E)$, with $|V|=n$ and $w \in V$.

The centrality measures that we study were designed to work with single vertices \cite{Freeman} rather than groups of them, thus we use their natural extension to sets as proposed in \cite{everett}.

Let $G=(V,E)$ be a graph and $S\subset V$. We denote by $C_j(S)$ for
$j \in \{d,c, bc\}$ the  \emph{group centrality} according to
the centrality measure $C_j$.

\emph{Group degree centrality} $C_d(S)$ counts the number of vertices not in $S$ that are connected to some vertex in $S$. Multiple ties to the same vertex are counted only once.
\emph{Group closeness} $C_c(S)$ studies the value of the sum of the distances from $S$ to all vertices not in it. Note that for group degree, the higher the value of $C_d(S)$, the more central the set is. Whereas for group closeness the reciprocate holds. Formally, they are defined as follows:
\begin{align*}
C_{d}(S)=\frac{|\{v \in V\setminus S, N(v)\cap S \neq \phi \}|}{|V\setminus S|}\ ,\hspace{1cm} C_c(S)=\frac{\sum_{v\in V\setminus S}d(v,S)}{|V\setminus S|}\ .
\end{align*}
%and \emph{group closeness} of $S$ are defined respectively as

%The \emph{normalized group closeness centrality} of $S$ is defined as
%\begin{align*}
%C_c(S)=\frac{\sum_{v\in V\setminus S}d(v,S)}{|V\setminus S|}.
%\end{align*}
\emph{Group betweenness centrality} $C_{bc}(S)$ indicates the proportion of geodesics connecting pairs of vertices (that are not in $S$), that pass through $S$
\begin{align*}
C_{bc}(S)=\frac{2BC(S)}{|V\setminus S|(|V\setminus S|-1)},\ \text{ with }\ BC(S)=\sum_{\substack{u<v\\u,v\notin S}}\frac{\sigma_{uv}(S)}{\sigma_{uv}},
\end{align*}
where $\sigma_{uv}$ is the total number of shortest paths from $u$ to $v$ and $\sigma_{uv}(S)$ is the number of those paths that pass through $S$. As for group degree, the higher the value of $C_{bc}$, the better.
%The closeness of $v$ is $1/\sum_{u \in V, u\not=v}d(u,v)$. Its betweenness is $\sum_{s\neqv\neqt \in V, u\not=v}\sigma_{st}(v)/\sigma_{st}$, where $\sigma_{st}$ is the total number of shortest paths from $s$ to $t$ and $\sigma_{st}(v)$ is the number of those paths that pass through $v$. Finally, the pagerank of $G$ is the unique vector $p$ satisfying $p=\alpha p \bar{A}+(1-\alpha)q$ where $\bar{A}$ is the $l_1$-normalized adjacency matrix of the graph, $\alpha \in [0,1)$ is a damping factor, and $d$ is a preference vector (which is a distribution over $V$). Therefore the pagerank score of $v$ is the component of $p$ associated to vertex $v$.

%------------------------------------------------------------------------------------------------------------------------------------------------------------------
%---------------------------------------------------------------------------------------------------------------------------------------------------------

\subsection{Random Walk Centrality}
\paragraph{Random Walks}
The next definitions come from the work of Lov\'asz in random walk theory \cite{Lovasz}.
\label{sec:Bs}
Let $G=(V,E)$ be a graph such that $|V|=n$ and $|E|=m$, where $n,m \in \mathbb{N}$. Formally, a \emph{random walk} is a sequence of vertices obtained as follows: it starts at a vertex $v_0$, and if at the $t$-th step it is at a vertex $v_t=u$, it moves to a neighbor $v$ of $u$ with probability $p_{uv}=1/d(u)$. %Note that the sequence of random vertices $(v_t: t=0,1,...)$ is a Markov chain.
%%%%%

%%%%%
%$P_t$ will denote the distribution of $v_t$: $P_t(v)=\mathbb{P}(v_t=v)$. The vertex $v_0$ may be fixed, but may also be drawn from an initial distribution $P_0$. This initial distribution is said to be $stationary$ if $P_1=P_0$ (which will imply that $P_t=P_0 \ \forall t \geq 0$, because of the construction of the random walk). It can be easily proved that the distribution $\pi(v):=\omega(v)/\sum_{v \in V}\omega(w)$ is stationary for every graph $G$ and weight function $\omega$. From now on $\pi$ will be referred simply as the \emph{stationary distribution} (it is not difficult to prove that this distribution is unique, which makes this reference valid).
\color{black}

%\begin{definition}
\normalfont Let $S \subseteq V$.
 The \emph{hitting time for a set} $H(u,S)$ is the expected number of steps that a random walk starting at vertex $u$ takes to reach some vertex in $S$ for the first time. When $S=\{v\}$ is a singleton, we will simply write $H(u,v)$.
%\end{definition}

%------------------------------------------------------------------------------------------------------------------------------------------------------------------

One of the issues that all three measures of group centrality (presented on 2.1) have in common, is that they only take into consideration shortest paths between pairs of vertices. A semantic network represents relations between concepts, motivated by an specific context. Should that context change, new and more direct connections may arise. In that regard, as much connections as possible should be taken into consideration, which is precisely what random walk centrality attempts to do.
\color{black}
\begin{definition}[Random Walk Centrality]
\normalfont The \emph{random walk centrality} of $S \subset V$ is
\[
h^{\swarrow}(S) = \displaystyle \frac{\sum_{v \in V\setminus S}H(v,S)}{|V\setminus S|}.
\]
\end{definition}

%\noindent For $S=V$, $h^{\swarrow}(S)=0$.\\

Essentially random walk centrality is a variation of closeness that takes into consideration all ways to reach $S$ from a vertex $v$, rather than only through the shortest path. As for group closeness, the lower the value of $h^{\swarrow}(S)$, the more central the set is.

\subsection{The arguments for the SWOT analysis}

Let us state them in general terms the arguments of the SWOT analysis.

\paragraph{Obstacles: weaknesses and threats.}
First of all, the problem is computationally hard;
we prove that for all four measures considered it is NP-hard:

\begin{theorem}
The problem of finding an optimal solution $S$ of size $k$, is an NP-hard problem for each of the four group centrality measures
studied in this paper \emph{(}that is, for $C \in \{C_d,C_c,C_{bc},h^{\swarrow}\}$\emph{)} .
\end{theorem}

This theorem points to the difficulties of applying the notion to big data
(and in this paper for experimenting with it) and gives a first program
of research on this problem: to develop good approximate algorithms for it.

Second, there is a difficulty associated to experimenting with big data due to the
lack of good benchmarks and semantically ``marked'' data, that is,
networks for which we have a clear agreement on what are their most central
concepts at different granularities (i.e. the size of an optimum candidate).  This is a challenge by itself (the scope of human
physical abilities)  and that is the reason
why we use small (human-scale) datasets in this paper.

 Third, the quality of data. It turns out that the output of the
 centrality and group centrality methods depends heavily on the ``quality'' of the network: in other words, in how well it represents the semantics
we are looking for. The examples we discussed on section 1 are prime cases
of excellent semantic networks (in fact, developed by psychologists based
on systematic experimentation). Unfortunately, the semantic networks that we see in real life
(for example RDF data) are usually not as good.
 We will show through several sample datasets that many of our current
most popular semantic data representations (e.g. RDF graphs) have a bias as they are
mostly constructed in ad-hoc manners (union of pieces
of information scattered from different sources; without the goal
of completeness, etc).\footnote{We are  well aware of the
complexities of determining in a network of
concepts which one is more relevant. This relies on the underlying
semantics of the notions involved, on the domain from which they are
taken from, on the goals that the people that built the network had,
and finally, on the quality of the network itself. Probably due to all of
the above,
today we do not have standard, nor good benchmark data for these tasks.
}

\paragraph{Strengths and Opportunities.}
  First, the notion of {\em group centrality} seems to be worth
exploring in the quest for developing summaries, small
representations or give insight about the meaning and contents of
big semantic network datasets. We show samples of different types of
networks, including stable RDF datasets like DBpedia. The results were
in all cases encouraging.

Second, the notion seems to be robust and become more stable as the
size of the optimum candidate $k$ increases. There are indications from the sample
data we worked out and partial theoretical results (Proposition \ref{bound}),
that in the long term, all measures behave in the same way as $k$ converges to $|V|/2$ (that is, there exists a set that is
optimum for all four measures). From this we have a second line of research:
to develop concluding analytical and real data studies that can support this conjecture.

 Third, the notion of {\em random walk} shows in general less variance
than other measures, i.e. gives more ``unique'' solutions for $k >0$.
%Also we have some partial theoretical evidence as shown by (Prop. \ref{poly}
%that bounds the notion of group centrality random walk by
%that of degree and closeness.
This virtue and the fact that this measure systematically
 gives solutions that have elements in common with those obtained through the use of other measures, are
good signals of its significance.
 If we link this observation to the fact that its theoretical basis
(the notion of random walk) has been successfully used in semantic
networks (see Related Work section) there is good ground to make
it a candidate for systematic study at the theoretical and experimental level.

\subsection{Related Work}
%The notion of group centrality seems to have not been systematically
%studied besides the work of Everett and Borgatti CITAR LOS TRES PAPERS DE ELLOS.

%{\em Random Walks and Semantic Networks.}

%The notion of random walk centrality seems to not have been systematically addressed, although the notion occurred in CITA!! Two of the authors have done some analytical studies
%on them CITAR  proving that it is indeed a ``good'' centrality measure according to the axioms of Boldi-Vigna CITAR

%%%%
{\em Random Walks and Semantic Networks.}
The motivation to address random walk centrality comes from works that have used the notion of random walks and related it with semantics. The work of Abbott et al. \cite{HumanMemory}  compares the functioning of the human mind when searching for memories with a random walk in a semantic network, as both present a similar behavior. They conclude that these results can help clarify the possible mechanisms that could account for PageRank predicting the prominence of words in semantic memory.
%%%%%%%%%%
Yeh et al. \cite{WikiWalk}  use random walks to determine the semantic relatedness of two elements. In particular, they use a personalized PageRank algorithm with a custom teleport vector. They conclude that random walk combined with personalized PageRank is a feasible and potentially fruitful mean of computing semantic relatedness for words and texts.
%%%%%%%%%
Hughes et al. \cite{LexicalSemantic} introduce a new measure of lexical relatedness based on the divergence of the stationary distributions computed using random walks over graphs extracted from WordNet.
 All these works have been valuable sources of inspiration for our method.

{\em Relevance and quality ranking in the Semantic Web.}
 As our motivation of studying group centrality points to its use in the
area of relevance in Semantic data, below we discuss some related work
made on it.
% subsection ranking_nodes_in_rdf
Graves et al.\cite{RankNodesRDF} propose a graph-theoretic measure
called \textit{noc-order} for ranking nodes in RDF graphs. %motivated by the notion of centrality.
They base their ranking in a variation of
\textit{closeness centrality}, which measures how costly is to reach
all the nodes in the graph from one specific vertex. Note that this
%is the reverse of $\bar{h}(S)$ and
idea does not use the concept
of expected length (via random walks).
% The centrality is equivalent to computing the \textit{All-pairs shortest path
%  problem}. In order to calculate the \textit{noc-order} ranking they
% define several measures of distance.
% First. they define $d_0(x,y,p) = m$ where m is the number of edges in the path $p$. Also, they use $d_w(x,y,p)=\sum_{i=1}^m w(e)$, this is the sum of the weights of the edges in the path based on a given weight function $w$. For this work, they consider the frequency of occurrence of each predicate for weights. Finally, they define $d_w^\alpha(x,y,p)=\sum_{i=1}^m(w(e)/\alpha^i)$ where the distance gets longer as a function of $\alpha$.
%%%%%%%%%%%%%%%%%%%%%%%%%%%%%%%%%%%%%%%%%%%%%
Zhang et al. \cite{zhang2007ontology} address the problem of
summarizing RDF graphs using an RDF sentence as the atomic
unit. They focus in ranking the RDF sentences to get a summary of the
ontology and base their method in centrality measures such as degree-,
eigenvector-, or betweenness-centrality.
%%%%%%%%%%%%
 Cheng et al.~\cite{Cheng2011} address the problem of
ranking features of an entity in a semantic graph. Their method
relies on centrality notions enhanced with the information captured by labels.
  The topics of these works is related to our work,
though we focus in finding central elements in the graph as a whole,
while  they do it in summarizing entities/ontologies; also
they do not present a conceptual notion of key elements,
rely only on experimental evaluations and do not address group centrality.
%%%%%%%%%
Ding et al. \cite{ding2005finding} address ranking of semantic objects
on the Web. They concentrate on documents, terms and RDF graphs as a
whole. For terms, they use popularity measures essentially based on
use and population, weighted by some parameters of the dataset
(occurrences, etc.).
%%%%%%%%%%%
Supekar et al. \cite{supekar2004characterizing} address the problem of
quality of ontologies. They characterize ontological features to
calculate a score in order to determine the quality of the ontology. They
characterize features like how well the ontology reflects the domain
of interest, correctness of the ontology with respect to the domain of
discourse, how easy the ontology is to understand or the depth and
width of the ontology.
%%%%%%%%%%%%
  These last two papers deal with ranking of ontologies as such,
but do not address the problem of finding a set of relevant concepts/nodes
in RDF graphs, nor present any conceptual framework to define them.

%%%%%%%%%%%%%%%%%%%%%%%%%%%%%%%%%%%%%%%%%%%%%%%%%%%%%%
%%%%%%%%%%%%%%%%%%%%%%%%%%%%%%%%%%%%%%%%%%%%%%%%%%%%%%
%%%%%%%%%%%%%%%%%%%%%%%%%%%%%%%%%%%%%%%%%%%%%%%%%%%%%%
%%%%%%%%%%%%%%%%%%%%%%%%%%%%%%%%%%%%%%%%%%%%%%%%%%%%%%
\section{Group Centrality on Semantic Network Data}

    We are interested on knowing what the group centrality approach applied to
    semantic networks has to offer. In order to do this, we apply group degree,
    closeness, betweenness, and random walk centrality notions to
    semantic networks data and study their behavior.
    As was stated on section 2,  we restrict our analysis to small networks  for
    two reasons: finding a set with optimal value for an specific
    centrality measure is NP-hard; and small networks (human scale) allow to
   get an insight about which concepts these group centralities select and
    how those selections differ from each other.

    \subsection{Datasets}

\noindent
     \textbf{Networks of natural concepts} \cite{proximity}.
This is the dataset we presented in the introduction.
There we showed two networks  built, respectively, by a group of
undergraduate psychology  and by graduate biologist.
Additionally, we show here two other networks.
A group of students from
        introductory psychology courses were asked to rate the degree of
        relatedness of all pairwise combinations of 25 natural concepts. Using
        this information the authors constructed two networks, according to two
        different thresholds of concept relatedness.

\noindent
     \textbf{Perception: semantic network of common sense} \cite{perception}.
        This network stores knowledge about what things look and feel like. It
        has about $9,000$ manually annotated relations between $4,000$ mundane
        concepts, e.g, cat, house, snow, etc.

\noindent
     \textbf{DBpedia Categories} \cite{DBpedia}.
        This dataset covers all the Wikipedia article categories and how they
        are related using the SKOS vocabulary. The network covers $1,130,342$
        categories with $2,308,585$ relations between them.

\noindent
     \textbf{Roget's Thesaurus}
        This is a widely used English-language thesaurus that contains more than
        $15,000$ words. This network was been used in works studying semantic
        similarity and sense disambiguation, among others.

\smallskip
    For each of the last 3 networks we extracted samples of 40 nodes
using  random-walk sampling according to
 Leskovec and  Faloutsos \cite{samplinglargegraphs} until 40 different nodes were visited and then
    we extracted the induced subgraph. Here we show two such samples
    for each dataset.

    \begin{table}[h!]
        \centering
        {
        \scriptsize
        \begin{tabularx}{\textwidth}{| c | Y | Y | Y || Y |}
            \hline
            ~n & degree & closeness & betweenness & random-walk \\ \hline
            1 & \{11\} & \{11\} & \{11\} & \{11\} \\ \hline
            2 & \{4, 11\} & \{2, 13\} , \{4, 11\} & \{4, 11\} & \{4, 11\} \\ \hline
            3 & \{4, 11, 13\} , \{0, 11, 13\} , ... (9) & \{0, 11, 13\} , \{4, 11, 13\} , ... (4) & \{4, 11, 13\} & \{0, 11, 13\} \\ \hline
            4 & \{2, 4, 11, 13\} , ... (23) & \{2, 4, 11, 13\} , ... (5)] & \{2, 4, 11, 13\} & \{2, 4, 11, 13\} \\ \hline
        \end{tabularx}
        }
        \smallskip
        \caption{
            Solutions for group centrality on Proximity Semantic Network 1.
        }
        \label{table:proximity}
    \end{table}
\vspace{-7mm}
    \begin{table}[h!]
        \centering
        {
        \scriptsize
        \begin{tabularx}{\textwidth}{| c | Y | Y | Y || Y |}
            \hline
~n~ & degree & closeness & betweenness & random-walk \\ \hline
1 & \{11\} & \{11\} & \{11\} & \{11\} \\ \hline
2 & \{0, 11\} & \{2, 11\} , ... (2) & \{2, 11\} & \{2, 11\} \\ \hline
3 & \{2, 13, 21\} , ... (8) & \{2, 13, 21\} & \{2, 11, 13\} & \{2, 13, 21\} \\ \hline
4 & \{2, 4, 13, 21\} , ... (16) & \{2, 4, 13, 21\} & \{2, 11, 13, 21\} & \{2, 4, 13, 21\} \\ \hline
        \end{tabularx}
        }
        \smallskip
        \caption{
            Solutions for group centrality on Proximity Semantic Network 2.
        }
        \label{table:proximity11}
    \end{table}
\vspace{-7mm}
    \begin{table}[h!]
        \centering
        {
        \scriptsize
        \begin{tabularx}{\textwidth}{| c | Y | Y | Y || Y |}
            \hline
~n~ & degree & closeness & betweenness & random-walk \\ \hline
1 & \{6\} , \{17\} & \{6\} & \{17\} & \{6\} \\ \hline
2 & \{0, 6\} & \{0, 6\} , \{10, 17\} & \{0, 6\} & \{0, 6\} \\ \hline
3 & \{3, 13, 17\} , \{3, 10, 17\}, ... (10) & \{3, 10, 17\} & \{0, 13, 17\} & \{3, 10, 17\} \\ \hline
4 & \{0, 6, 10, 17\} , ... (18) & \{0, 6, 10, 17\} , ... (6) & \{0, 6, 10, 17\} & \{0, 6, 10, 17\} \\ \hline
        \end{tabularx}
        }
        \smallskip
        \caption{
          Solutions for group centrality on Proximity Semantic Network constructed by non-biologist students.
        }
        \label{table:proximity2}
    \end{table}
\vspace{3mm}
    \begin{table}[h!]
        \centering
        {
        \scriptsize
        \begin{tabular}{| c | c | c | c || c |}
            \hline
            ~n~ & degree & closeness & betweenness & random-walk \\ \hline
1 & \{12\} & \{12\} & \{12\} & \{12\} \\ \hline
2 & \{3, 12\} & \{3, 12\} & \{3, 12\} & \{3, 12\} \\ \hline
3 & \{3, 12, 19\}, \{3, 12, 21\} & \{3, 12, 21\}, \{3, 12, 22\} & \{3, 12, 19\}, \{3, 12, 21\} & \{3, 12, 22\} \\
 & \{3, 12, 22\}, \{3, 12, 23\} & \{3, 12, 23\}  &  &  \\ \hline
4 & \{3, 12, 19, 23\} & \{3, 12, 19, 23\} & \{3, 12, 15, 19\}, \{3, 12, 15, 21\} & \{3, 5, 12, 22\} \\ \hline

        \end{tabular}
        }
        \smallskip
        \caption{
          Solutions for group centrality on Proximity Semantic Network constructed by biologist students.
        }
        \label{table:proximity3}
    \end{table}
\vspace{-7mm}
    \begin{table}[h]
        \centering
        {
        \scriptsize
        \begin{tabular}{| c | c | c | c || c |}
            \hline
~n~ & degree & closeness & betweenness & random-walk \\ \hline
1 & \{4\} , \{20\} & \{4\} & \{20\} & \{4\} \\ \hline
2 & \{4, 20\} & \{4, 20\} & \{4, 20\} & \{4, 20\} \\ \hline
3 & \{4, 6, 20\} , \{6, 7, 20\} & \{0, 4, 20\} & \{4, 7, 20\} & \{4, 18, 20\} \\ \hline
4 & \{4, 6, 18, 20\} , \{0, 4, 6, 20\} & \{0, 4, 6, 20\} & \{4, 6, 7, 20\} & \{4, 6, 18, 20\} \\ \hline
        \end{tabular}
        }
        \smallskip
        \caption{
            Solutions for group centrality on DBpedia sample 1.
        }
        \label{table:dbpedia1}
    \end{table}
\vspace{-7mm}
    \begin{table}[h!]
        \centering
        {
        \scriptsize
        \begin{tabular}{| c | c | c | c || c |}
            \hline
~n~ & degree & closeness & betweenness & random-walk \\ \hline
1 & \{13\} & \{13\} & \{13\} & \{13\} \\ \hline
2 & \{0, 13\} , \{1, 13\} , ... (2) & \{1, 13\} & \{0, 13\} & \{12, 13\} \\ \hline
3 & \{0, 1, 13\} , ... (4) & \{0, 1, 13\} & \{0, 1, 13\} & \{0, 1, 13\} \\ \hline
4 & \{0, 1, 3, 13\} , ... (1) & \{0, 1, 3, 13\} & \{0, 1, 3, 13\} & \{0, 1, 3, 13\} \\ \hline
        \end{tabular}
        }
        \smallskip
        \caption{
            Solutions for group centrality on DBpedia sample 2.
        }
        \label{table:dbpedia2}
    \end{table}

%     \begin{table}[h]
%         \centering
%         {
%         \scriptsize
%         \begin{tabular}{| c | c | c | c || c |}
%             \hline
% ~n~ & degree & closeness & betweenness & random-walk \\ \hline
% 1 & \{5\} & \{1\} & \{1\} & \{1\} \\ \hline
% 2 & \{4, 5\} & \{1, 5\} & \{4, 5\} & \{1, 5\} \\ \hline
% 3 & \{4, 5, 23\} & \{4, 5, 23\} & \{4, 5, 23\} & \{4, 5, 23\} \\ \hline
% 4 & \{4, 5, 20, 23\} & \{4, 5, 20, 23\} & \{4, 5, 20, 23\} & \{4, 5, 20, 23\} \\ \hline
%         \end{tabular}
%         }
%         \smallskip
%         \caption{
%             Solutions for group centrality on group Centrality on DBpedia sample 3.
%         }
%         \label{table:dbpedia3}
%     \end{table}

    \begin{table}[h!]
        \centering
        {
        \scriptsize
        \begin{tabular}{| c | c | c | c || c |}
            \hline
~n~ & degree & closeness & betweenness & random-walk \\ \hline
1 & \{25\} & \{14\} & \{14\} & \{13\} \\ \hline
2 & \{14, 25\} & \{5, 14\} , \{7, 13\} , \{7, 14\} , \{8, 13\} & \{14, 25\} & \{6, 14\} \\ \hline
3 & \{4, 14, 25\} , ... (5) & \{4, 14, 25\} , ... (4) & \{4, 14, 25\} & \{0, 6, 14\} \\ \hline
4 & \{0, 4, 14, 25\} , ... (10) & \{1, 4, 14, 25\} , ... (2) & \{4, 12, 14, 25\} & \{0, 4, 14, 25\} \\ \hline
        \end{tabular}
        }
        \smallskip
        \caption{
            Solutions for group centrality on Perception sample 1.
        }
        \label{table:perception1}
    \end{table}

    \begin{table}[h!]
        \centering
        {
        \scriptsize
        \begin{tabular}{| c | c | c | c || c |}
            \hline
~n~ & degree & closeness & betweenness & random-walk \\ \hline
1 & \{29\} , ... (1) & \{15\} & \{1\} & \{29\} \\ \hline
2 & \{2, 29\} & \{0, 15\} , ... (1) & \{1, 2\} & \{29, 39\} \\ \hline
3 & \{2, 8, 29\} , ... (3) & \{0, 15, 32\} , ... (7) & \{1, 2, 8\} & \{2, 8, 39\} \\ \hline
4 & \{2, 8, 29, 39\} , ... (2) & \{2, 26, 29, 39\} , ... (1) & \{2, 8, 10, 29\} & \{2, 26, 29, 39\} \\ \hline
        \end{tabular}
        }
        \smallskip
        \caption{
            Solutions for group centrality on Perception sample 2.
        }
        \label{table:perception2}
    \end{table}

%     \begin{table}[h]
%         \centering
%         {
%         \scriptsize
%         \begin{tabular}{| c | c | c | c || c |}
%             \hline
% ~n~ & degree & closeness & betweenness & random-walk \\ \hline
% 1 & \{11\} & \{11\} & \{11\} & \{11\} \\ \hline
% 2 & \{4, 11\} & \{2, 13\} , \{4, 11\} & \{4, 11\} & \{4, 11\} \\ \hline
% 3 & \{0, 11, 13\} , \{4, 11, 13\} , ... (9) & \{0, 11, 13\} , \{4, 11, 13\} , ... (4) & \{4, 11, 13\} & \{0, 11, 13\} \\ \hline
% 4 & \{2, 4, 11, 13\} , ... (23) & \{2, 4, 11, 13\} , ... (5) & \{2, 4, 11, 13\} & \{2, 4, 11, 13\} \\ \hline
%         \end{tabular}
%         }
%         \smallskip
%         \caption{
%             Solutions for group centrality on group Centrality on Perception sample 3.
%         }
%         \label{table:perception3}
%     \end{table}

    \begin{table}[h!]
        \centering
        {
        \scriptsize
        \begin{tabular}{| c | c | c | c || c |}
            \hline
~n~ & degree & closeness & betweenness & random-walk \\ \hline
1 & \{10\} & \{35\} & \{35\} & \{10\} \\ \hline
2 & \{10, 25\} & \{10, 25\} & \{25, 35\} & \{10, 25\} \\ \hline
3 & \{10, 25, 39\} & \{10, 25, 39\} & \{10, 25, 35\} & \{10, 25, 39\} \\ \hline
4 & \{10, 17, 25, 39\} , ... (87) & \{10, 17, 25, 39\} , ... (87) & \{10, 17, 25, 35\} & \{10, 17, 25, 39\} \\ \hline
        \end{tabular}
        }
        \smallskip
        \caption{
            Solutions for group centrality on Roget's Thesaurus sample 1.
        }
        \label{table:thesaurus1}
    \end{table}

%     \begin{table}[h]
%         \centering
%         {
%         \scriptsize
%         \begin{tabular}{| c | c | c | c || c |}
%             \hline
% ~n~ & degree & closeness & betweenness & random-walk \\ \hline
% 1 & \{29\} & \{9\} & \{9\} & \{29\} \\ \hline
% 2 & \{26, 29\} & \{26, 29\} & \{9, 29\} & \{26, 29\} \\ \hline
% 3 & \{10, 12, 26\} & \{10, 12, 26\} & \{9, 26, 29\} & \{10, 12, 26\} \\ \hline
% 4 & \{10, 12, 26, 32\} , ... (4) & \{10, 12, 24, 26\} , ... (4) & \{10, 12, 26, 29\} & \{10, 12, 26, 29\} \\ \hline
%         \end{tabular}
%         }
%         \smallskip
%         \caption{
%             Solutions for group centrality on group Centrality on Roget's Thesaurus sample 2.
%         }
%         \label{table:thesaurus2}
%     \end{table}

    \begin{table}[h!]
        \centering
        {
        \scriptsize
        \begin{tabular}{| c | c | c | c || c |}
            \hline
~n~ & degree & closeness & betweenness & random-walk \\ \hline
1 & \{11\} & \{37\} & \{37\} & \{37\} \\ \hline
2 & \{11, 14\} & \{14, 26\} & \{26, 37\} & \{14, 26\} \\ \hline
3 & \{1, 14, 29\} , ... (6) & \{11, 14, 29\} , ... (6) & \{11, 26, 37\} & \{11, 14, 29\} \\ \hline
4 & \{11, 14, 29, 37\} , ... (440) & \{11, 14, 29, 37\} , ... (440) & \{11, 14, 26, 37\} & \{11, 14, 29, 37\} \\ \hline
        \end{tabular}
        }
        \smallskip
        \caption{
            Solutions for group centrality on Roget's Thesaurus sample 2.
            %this is thesaurus 3
        }
        \label{table:thesaurus3}
    \end{table}

\newpage
    \subsection{Discussion of results}

    % In the following tables we present the concepts captured by the group
    % centrality measures on the selected semantic networks. We selected groups
    % of size 1, 2, 3, and 4, since larger groups took too long to compute. For
    % each semantic network we present a table showing which group of nodes are
    % captured by each centrality measure and each group size. A group
    % centrality measure can capture more than one optimal group, but in the
    % tables we only show two of the most frequent options between de measures.
    % To observe the complete results, the graphs, and the labels of the
    % concepts please refer to
    % \url{http://dcc.uchile.cl/~cgarrido/research/groupcentrality/results}.

    For each semantic network, and each centrality measure, we computed the
    optimal solutions of size 1, 2, 3, and 4. The results are presented in
    Tables 1 to 10. It is important to note that in some cases there were
    multiple candidates for optimal solution, but we only include at most two of
    them\footnote{ The complete results, graphs and labels of the networks can
    be found in the appendix.} (if
    there are more than two solutions, there is a  mark ``... $(n)$'' where $n$
    is the number of extra solutions) .

    From the results obtained, we can advance two main
    observations:

\noindent
 {\em Robustness of group centrality.} It is noticeable that in the samples
 the optimal solutions for all centrality measures become
 similar as the size $k$ of the set increases. Indeed, for groups
 of size 1, it can observed
    that there are networks for which the solution is different, depending on the
    centrality notion applied (e.g. Table \ref{table:proximity2}). On the other hand, for groups of size
    4, it can be observed that in most networks there is a common set of nodes
    selected by all four different measures used.

\smallskip
\noindent
{\em Non-ambiguity of group random-walk.} As stated previously, for a
given size $k$, more than one optimal solution may exist for group centrality. Moreover, from the results
 it can be inferred that as $k$ increases, so does the number of optimal candidates for
    certain particular measures. For example, Table
    \ref{table:proximity} shows that degree and
    closeness centrality have 24 (resp. 6) optimal solutions of size 4.
    There is even an extreme case in Table \ref{table:thesaurus3}, where degree and
    closeness have more than 400 different optimal group solutions.
   In a small sample of 40 nodes this variance does not help achieve the proposed goal.
 However, betweenness centrality seems to behave in general better in this
 matter (there is some dispersion, e.g.  Table \ref{table:proximity3},
 but it is small). Random-walk centrality, on the other hand, is consistently
better in this regard. In our sample selections, it is the only notion
for which there is systematically a unique optimal solution for any size, and for all the networks.

%%%%%%%%%%%%%%%%%%%%%%%%%%%%%%%%%%%%%%%%%%%%%%%%%%%%%%%%%%%%%%%%%%%
%%%%%%%%%%%%%%%%%%%%%%%%%%%%%%%%%%%%%%%%%%%%%%%%%%%%%%%%%%%%%%%%%%%
%%%%%%%%%%%%%%%%%%%%%%%%%%%%%%%%%%%%%%%%%%%%%%%%%%%%%%%%%%%%%%%%%%%
%%%%%%%%%%%%%%%%%%%%%%%%%%%%%%%%%%%%%%%%%%%%%%%%%%%%%%%%%%%%%%%%%%%
%%%%%%%%%%%%%%%%%%%%%%%%%%%%%%%%%%%%%%%%%%%%%%%%%%%%%%%%%%%%%%%%%%%
\section{Theoretical aspects of Group Centrality}

\color{black}

As we saw in previous sections, group centrality offers interesting
new insights on semantic networks. In this section we study the notion
from a theoretical point of view, focusing in the complexity of
computing it.

\subsection{Complexity of Group Centrality}

Let us define formally the problem of computing the most central set
of nodes in a network under different centrality measures.

\begin{definition}[Optimumset]
Given a graph $G=(V,E)$, a positive integer $k$ and a function $C:\mathcal{P}(V) \rightarrow \mathbb{R}$, define $optimumset(G,k,C)$ as the problem of finding a set $S\subset V$, of size $k$ such that
\[
S \in \displaystyle \argmin_{U \subset V;|U|=k}C(S).
\]
For a real number $\alpha>0$, $optimumset(G,k,C,\alpha)$ will denote the corresponding decision problem:  find if there exists a set $S\subset V$, of size $k$ such that $C(S)=\alpha$.
\end{definition}
%\begin{theorem}
%Given a graph $G=(V,E)$, a positive integer $k$, the corresponding decision problems of finding if there exist a subset $S \subset V$ of size $k$ such that its group centrality is optimum is NP-complete for all four group centrality measures degree, closeness, betweenness and random walk.
%\end{theorem}
%\begin{theorem}
%$optimumset(G,k,C,\alpha)$ is NP-complete for all four group centrality measures studied in this paper \emph{(}that is, for $C \in \{C_d,C_c,C_{bc},h^{\swarrow}\}$\emph{)}.
%\end{theorem}
\textbf{Theorem 2.}\ \ \emph{$optimumset(G,k,C,\alpha)$ is NP-complete for all four group centrality measures studied in this paper \emph{(}that is, for $C \in \{C_d,C_c,C_{bc},h^{\swarrow}\}$\emph{)}.}

\begin{proof}
First note that it is not difficult to prove that $C_d$, $C_c$ and $C_{bc}$ are all computable in polynomial time. Same for $h^{\swarrow}$ but we will leave that demonstration for later (see proposition \ref{poly}). Thus the corresponding problems are all in NP.

Now, consider $\alpha=1$ and let $S_d$, $S_c$, $S_{bc}$ and $S_{h^{\swarrow}}$ be solutions of the corresponding decision problems.

Note that $C_c(S_c)=1$ (resp. $C_d(S_d)=1$) if and only if $S_c$ (resp. $S_d$) is a dominating set of cardinality $k$ in $G$. That is, the problem of finding a dominating set of cardinality $k$ in $G$ is polynomial time reducible to $optimumset(G,k,\alpha,C_c)$ and $optimumset(G,k,\alpha,C_d)$.

Indeed, $C_c(S_c)=1$ if and only if $d(v,S_c)=1, \forall v \notin S_c$, which happens if and only if every vertex $v \notin S_c$ has at least one neighbor on $S_c$ (which is also a necessary and sufficient condition for $C_d(S_d)=1$). That is, $S_c$ (and also $S_d$) is a dominating set.

On the other hand, $C_{bc}(S_{bc})=1$ (resp. $C_{h^{\swarrow}}(S_{h^{\swarrow}})=1$) if and only if $S_{bc}$ (resp. $S_{h^{\swarrow}}$) is a vertex cover of $G$ of size $k$. That is, the problem of finding a vertex cover of $G$ of size $k$ is polynomial time reducible to $optimumset(G,k,\alpha,C_{h^{\swarrow}})$ and $optimumset(G,k,\alpha,C_{bc})$.

Indeed, $C_{bc}(S_{bc})=1$ if and only if
$\sigma_{uv}(S_{bc})=\sigma_{uv}, \forall u,v \notin S_{bc}$ with
$u<v$, which occurs if and only if for every pair $\forall u,v \notin
S_{bc}$ with $u<v$, all shortest paths that connect $u$ and $v$ have
some vertex in $S_{bc}$. This is equivalent to having $N(v) \subseteq
S_{bc}$ for all $v \in S_{bc}$. That is, $S_{bc}$ is a vertex cover of $G$.

Finally, $C_{h^{\swarrow}}(S_{h^{\swarrow}})=1$ if and only if $h(v,S_{h^{\swarrow}})=1, \forall v \notin S_{h^{\swarrow}}$, but $h(v,S_{h^{\swarrow}})=1$ occurs if and only if $N(v) \subset S, \forall v \notin S_{h^{\swarrow}}$. That is, $S_{h^{\swarrow}}$ is a vertex cover of $G$.

\qed
\end{proof}

\subsection{Properties of Random Walk Centrality}

\noindent Given a graph $G=(V,E)$,
we can generalize the notion of \emph{random walk} to the case where
there is a  weight function $\omega: E \rightarrow \mathbb{R}$
on the edges of $G$, by changing the transition probabilities for
 $p_{uv}=\omega(uv)/\omega(u)$, with $\omega(u)=\sum_{w \in
   N(u)}\omega(uw)$ (thus the case $\omega \equiv 1$ gives
 the original definition of random walk).

Note that the sequence of random vertices $(v_t: t=0,1,...)$ is a Markov chain. Let $P_t$ will denote the distribution of $v_t$: $P_t(v)=\mathbb{P}(v_t=v)$. The vertex $v_0$ may be fixed, but may also be drawn from an initial distribution $P_0$. This initial distribution is said to be $stationary$ if $P_1=P_0$ (which will imply that $P_t=P_0 \ \forall t \geq 0$, because of the construction of the random walk). It can be easily proved that the distribution $\pi(v):=\omega(v)/\sum_{v \in V}\omega(w)$ is stationary for every graph $G$ and weight function $\omega$. From now on $\pi$ will be referred simply as the \emph{stationary distribution} (it is not difficult to prove that this distribution is unique, which makes this reference valid).

\begin{proposition}
\label{poly}
Given  $G=(V,E)$, a weight function $\omega: E \rightarrow \mathbb{R}$,
a subset $S \subset V$ and $v\in V$, then  $h(v,S)$ and
$h^{\swarrow}(S) $ can be computed in polynomial time.
\end{proposition}

For the proof, first we need the following result

\begin{lemma}
\label{np}
For a graph $G=(V,E)$, a weight function $\omega$ and $S \subseteq V$ define $G_{\scriptscriptstyle S}$ as the following weighted graph:

$V(G_{\scriptscriptstyle S}) = (V \setminus S) \cup
\{v_{\scriptscriptstyle S}\}$

$E(G_{\scriptscriptstyle S}) = \{e \in E / e \cap S = \phi \} \cup \{e=uv_{\scriptscriptstyle S} / u \in V, N_{\scriptscriptstyle G}(u) \cap S \neq \phi\}$

\noindent
and $\omega_{\scriptscriptstyle S}:\ E(G_{\scriptscriptstyle S})
 \longrightarrow \mathbb{N}$ defined as:
			 $\omega_{\scriptscriptstyle S}(uv) = \omega(uv)$ if  $uv
                           \cap S = \emptyset$.
Otherwise, i.e. if $ v=v_{\scriptscriptstyle S}$, then
$ \omega_{\scriptscriptstyle S}(uv) =\sum\limits_{w \in N_{\scriptscriptstyle G}(u) \cap S}\omega(uw)$.

Then for $u \in V \setminus S$ we have that
$H_{\scriptscriptstyle G}(u,S)=H_{{\scriptscriptstyle G}_{\scriptscriptstyle S}}(u,v_{\scriptscriptstyle S}).$
\end{lemma}

\begin{proof}
For a vertex $w \in V \setminus S$, let $T_{w,S}$ and $T_{w,v_S}$ be two random variables that count the number of steps that a random walk (on $G$ and $G_S$ respectively) starting at vertex $w$ takes in order to reach some vertex of $S$ and $v_{\scriptscriptstyle S}$ (respectively) for the first time.
We will prove that $T_{u,S}$ and $T_{u,v_S}$ have the same distribution.

Let us proceed by induction on the possible values of $T_{u,S}$ and $T_{u,v_{\scriptscriptstyle S}}$.

If $n=1$ then
\begin{align*}
\mathbb{P}(T_{u,S}=1)=\frac{\sum\limits_{w \in N_G(u) \cap S}\omega(uw)}{\omega(u)}=\frac{\omega_{\scriptscriptstyle S}(uv_{\scriptscriptstyle S})}{\omega_{\scriptscriptstyle S}(u)}=\mathbb{P}(T_{u,v_{\scriptscriptstyle S}}=1).
\end{align*}
Let $n>1$ and assume that the property holds for $n-1$. Then
\begin{align*}
\mathbb{P}(T_{u,S}=n)&=\sum_{w \in V\setminus S}\frac{\omega(uw)}{\omega(u)}\mathbb{P}(T_{w,S}=n-1)\\
                     &=\sum_{w \in V\setminus S}\frac{\omega(uw)}{\omega(u)}\mathbb{P}(T_{w,v_S}=n-1)\\
										 &=\sum_{w \in V(G_S)\setminus v_S}\frac{\omega(uw)}{\omega(u)}(T_{w,v_S}=n-1)\\
										 &=\sum_{w \in V(G_S)\setminus v_S}\frac{\omega_{\scriptscriptstyle S}(uw)}{\omega_{\scriptscriptstyle S}(u)}\mathbb{P}(T_{w,v_S}=n-1)=\mathbb{P}(T_{u,v_S}=n).
\end{align*}
Using this we have that
$ H_{\scriptscriptstyle G}(u,S)=\mathbb{E}(T_{u,S})=\mathbb{E}(T_{u,v_S})=H_{{\scriptscriptstyle G}_S}(u,v_S).$
\qed
\end{proof}

\begin{proof}[of proposition \ref{np}]
Let $P$ be the probability transition matrix associated to a random
walk on $G$, and the function $\omega$. That is, $P_{ij}=\frac{\omega(ij)}{\omega(i)}$. Define $P_{\infty}$ as
\begin{align*}
P_{\infty}&=\begin{pmatrix}
\pi(1) & \pi(2) & \dots & \pi(n)\\
\vdots & \vdots & \dots & \vdots\\
\pi(1) & \pi(2) & \dots & \pi(n)
\end{pmatrix}.
\end{align*}
where $\pi$ is the stationary distribution. Define also $Z:=(I-(P-P_{\infty}))^{-1}$ and let $H$ be the hitting time matrix of $G$, that is, $H_{ij}=H_{\scriptscriptstyle G}(i,j)$. Then it can be proved as stated in \cite{zeta} that for $i,j \in V$
\begin{equation}
H_{\scriptscriptstyle G}(i,j)=\frac{Z_{\scriptstyle{jj}}-Z_{\scriptstyle{ij}}}{\pi(j)}. \label{eq1}
\end{equation}

%This tell us that in order to compute $H(i,j)$ for all pair of vertices $i,j$ in the graph, we only need to compute the matrix $Z$ which can be made in polynomial time.

Therefore, by combining equation (\ref{eq1}) and lemma \ref{np} we can
compute $h(v,S)$ by inverting a suitable matrix $Z$, which can be made in polynomial time.
\qed
\end{proof}

\begin{lemma}
There is a parameterized family of graphs for which
Random walk centrality behaves systematically differently than
degree, closeness and betweenness group centrality.
\end{lemma}

\begin{proof}
% One of the virtues of this notion is that it captures the informational flow of information that we discussed in the introduction.
%Consider $\{G_n^m\}$ Take for instance the graph $G$ in Fig. \ref{allbadbutone}. For $i \in \{1,...,m-1\}$, $K_n^i$ is a subgraph of $G$ that is a $n$-clique and that is connected to $v$ through an unique vertex $k_n^i$. On the other hand, $T_n^i$ is a subtree of $G$ of height 1 and size $n$ with root $t_i^n$ the only vertex in $T_n^i$ that is neighbor of $v$ in $G$. Suppose we want to find a set $S$ of size $m+1$ that minimizes some measure of centrality. It is not difficult to prove that for group closeness and betwenness centrality the sets $S_K=\{v,k_n^1,...,k_n^m\}$ and $S_T=\{v,t_n^1,...,t_n^m\}$ are solutions with the same centrality score. For group degree centrality $S_K=\{t_n^j,k_n^1,...,k_n^m\}$ and $S_T=\{k_n^j,t_n^1,...,t_n^m\}$ for any $1\leq j \leq m$ are also solutions with the same property. However, the difference in connectivity between sets $S_K$ and $S_T$ should be accounted for: one could argue that those connections represent extra channels through which information may flow. Here is precisely where our measure differentiates itself from the others, telling us that $S_K$ is more central than $S_T$.

Consider $\{G_n^m\}$, with $n,m \in \mathbb{Z}^{+}$ and $n<<m$ a family of graphs constructed as follows (see Fig. \ref{allbadbutone}):  for $i \in \{1,...,m-1\}$, $K_n^i$ is a subgraph of $G_n^m$ that is a $n$-clique and that is connected to $v$ through an unique vertex $k_n^i$. On the other hand, $T_n^i$ is a subtree of $G_n^m$ of height 1 and size $n$ with root $t_i^n$ the only vertex in $T_n^i$ that is neighbor of $v$ in $G_n^m$. Suppose we want to find a set $S$ that solves $optimumset(G_n^m,m+1,C)$ for $C \in \{C_d,C_c,C_{bc},h^{\swarrow}\}$. It is not difficult to prove that for group closeness and betweenness centrality the sets $S_K=\{v,k_n^1,...,k_n^m\}$ and $S_T=\{v,t_n^1,...,t_n^m\}$ are solutions. For group degree centrality $\bar{S}_K=\{t_n^j,k_n^1,...,k_n^m\}$ and $\bar{S}_T=\{k_n^j,t_n^1,...,t_n^m\}$ for any $1\leq j \leq m$ are possible solutions. However, one could argue the difference in connectivity between sets $S_K$ and $S_T$ should be accounted for. Indeed, for random walk centrality, $S_k$ (a more connected version of $S_t$) is the only solution.
\qed

%However, the difference in connectivity between sets $S_K$ and $S_T$ should be accounted for: one could argue that those connections represent extra channels through which information may flow. Here is precisely where our measure differentiates itself from the others, telling us that $S_K$ is more central than $S_T$.

%Take for instance the graph $G$ in Fig. \ref{allbadbutone}. For $i \in \{1,...,m-1\}$, $K_n^i$ is a subgraph of $G$ that is a $n$-clique and that is connected to $v$ through an unique vertex $k_n^i$. On the other hand, $T_n^i$ is a subtree of $G$ of height 1 and size $n$ with root $t_i^n$ the only vertex in $T_n^i$ that is neighbor of $v$ in $G$. Suppose we want to find a set $S$ of size $m+1$ that minimizes some measure of centrality.

%It is not difficult to prove that for group closeness and betwenness centrality the sets $S_K=\{v,k_n^1,...,k_n^m\}$ and $S_T=\{v,t_n^1,...,t_n^m\}$ are solutions with the same centrality score. For group degree centrality $S_K=\{t_n^j,k_n^1,...,k_n^m\}$ and $S_T=\{k_n^j,t_n^1,...,t_n^m\}$ for any $1\leq j \leq m$ are also solutions with the same property. However, the difference in connectivity between sets $S_K$ and $S_T$ should be accounted for: one could argue that those connections represent extra channels through which information may flow. Here is precisely where our measure differentiates itself from the others, telling us that $S_K$ is more central than $S_T$.
\end{proof}

%This is precisely the problem: all three measures are based on shortest distances between vertices.  as our measure does, telling us that the first one is most central.

% The classic measures of centrality (degree, betweenness and closeness) assign the same score to both $u$ and $v$ (pagerank even assigns a higher score to $v$). They do not give any importance to all the extra connections that are on the left side of $u$, connections that present extra opportunities for the information to travel. Note that in this case our notion of centrality correctly tells us that vertex $u$ is more important than $v$.

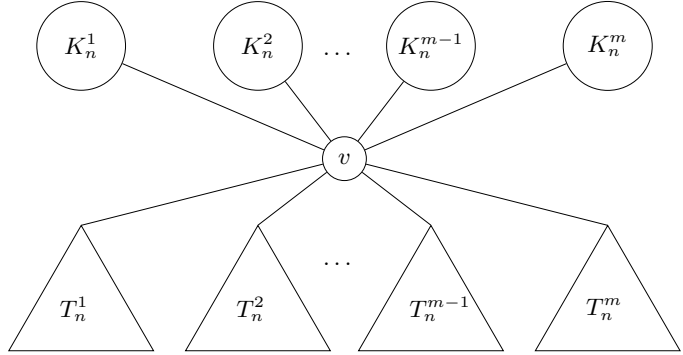
\begin{figure}[!h]
\centering
\begin{tikzpicture}[scale=1, every node/.style={scale=1}, line width=0.25pt,triangle/.style = {fill=white!20, regular polygon, regular polygon sides=3},node rotated/.style = {rotate=180},border rotated/.style = {shape border rotate=180}]
  \tikzstyle{vertex}=[align=center, inner sep=0pt, text centered, circle, black, draw=black, text width=3em]
	
	\node[vertex,text width=1.75em] (W-1) at (0,0) {$v$};
	\node[triangle,draw=black,text width=1.75em] (W-2) at (-3.5,-2) {$T_n^1$};
	\node[triangle,draw=black,text width=1.75em] (W-3) at (-1.15,-2) {$T_n^2$};
	\node[draw=white,text width=1.75em] (W-4) at (0,-1.4) {$\hdots$};
	\node[triangle,draw=black,text width=1.75em] (W-5) at (1.15,-2) {$T_n^{m-1}$};
	\node[triangle,draw=black,text width=1.75em] (W-6) at (3.5,-2) {$T_n^m$};
	\node[vertex,text width=3.5em] (W-7) at (-3.5,1.5) {$K_n^1$};
	\node[vertex,text width=3.5em] (W-8) at (-1.15,1.5) {$K_n^2$};
	\node[vertex,text width=3.5em] (W-9) at (1.15,1.5) {$K_n^{m-1}$};
	\node[vertex,text width=3.54em] (W-10) at (3.5,1.5) {$K_n^m$};
	\node[draw=white,text width=1.75em] (W-11) at (0,1.4) {$\hdots$};
	\node[] (W-a) at (-3.61,-0.915) {$$};
	\node[] (W-b) at (-1.26,-0.97) {$$};
	\node[] (W-c) at (3.61,-0.915) {$$};
	\node[] (W-d) at (1.26,-0.97) {$$};
	
  \foreach \from/\to in {1/a,1/b,1/c,1/d,1/7,1/8,1/9,1/10}
    \draw (W-\from) -- (W-\to);

\end{tikzpicture}
\caption{Which set of size $m+1$ is more central? $S_K$ or $S_T$?}
\label{allbadbutone}
\end{figure}

%\begin{theorem}
%\emph{Given a graph $G=(V,E)$, a positive integer $k<|V|$ and a real number $\alpha>0$, the problem of finding a subset $S\subset V$, of size $k$ such that $h^{\swarrow}(S)\leq \alpha$, is NP-complete.}
%\end{theorem}

%In the rest of the paper we will deal with the case of trees (for which we will present precise algorithm to compute the notion), and reduce the general case to this one, by using tree decompositions.

The following bound shows an interesting relation between random walk centrality, degree and closeness

\begin{proposition}
\label{bound}
 There exists a real constant $c>0$ such that for $S \subset V$
\[
h^{\swarrow}(S) \; \leq \;  \frac{1}{|V\setminus S|}\left(\sum_{u \in V\setminus S} d(u,S)\right)
                \left(\sum_{v \in V\setminus S} d(v)\right)  \leq c|V\setminus S|^3.
\]
\end{proposition}

\begin{proof}
Define $\partial(S):= \{v \in S: \exists u \in V \setminus S, uv \in E \}$. Define the subgraph $G_S$ of $G$ as follows:
$ V(G_S) := (V \setminus S) \cup \partial(S)$ and $E(G_S) := \{ uv \in E: u,v  \in V(G_S) \wedge u \in V \setminus S \} $.

First note that for any pair of vertices $u \in V \setminus S$ and $v \in \partial(S)$
\[
H_G(u,S) = H_{G_S}(u,\partial(S)) \leq H_{G_S}(u,v).
\]
The equality follows from the fact that the subgraph traversed and the target hit in both expressions are the same, and the inequality holds because the target of the expression on the right is more difficult to hit, thus needs more walking around.

Using this, we have that
\[
h^{\swarrow}(S)= \frac{\sum_{u \in V \setminus S} H_G(u,S)}{|V \setminus S|}
   \leq \frac{\sum_{u \in V \setminus S} H_{G_S}(u,v_u)}{|V \setminus S|},
\]
 with each $v_u \in \partial(S)$ chosen such that $d_{G_S}(u,v_u) = d_{G}(u,S)$.

On the other hand, by using the known fact that for any graph $L$ and any pair of vertices $u,v \in V(L)$, $H_L(u,v) \leq 2|E(L)|d_L(u,v)$, we have that
\begin{align*}
\frac{\sum_{u \in V \setminus S} H_{G_S}(u,v)}{|V \setminus S|} &\leq
                                                                  \frac{\sum_{u
                                                                  \in
                                                                  V
                                                                  \setminus S} 2|E(G_S)| d_{G_S}(u,v_u)}{|V\setminus S|}\\														&
\leq \frac{1}{|V\setminus S|}\left(\sum_{u \in V\setminus
                                                                                                                                                                                                                                          S}d_G(u,S)\right)\left(\sum_{v \in V\setminus S}d_G(v) \right).  \;\qed
\end{align*}
%\noindent The second inequality in the theorem follows directly from the first one by bounding the distance and degree of vertices for their maximum possible value over $V \setminus S$.
%\qed
\end{proof}

\section{Conclusions}
 We perform a SWOT analysis of the application of the notion of
group centrality to semantic networks, putting a particular
focus on random walk group centrality.
The challenges faced were mainly of two types. First,
the ambiguity of the very notion of the
``most central''  concept in a semantic network;
and second, the size of the networks.

 In spite of these obstacles (that we discussed in the paper), we hope we
presented enough evidence that this topic  deserves to be explored and
researched. We approached these obstacles by working with small samples and
advancing theoretical results.  Besides stating the problem at the formal level,
we proved that the problem of computing group centrality is NP-hard, and showed
that random walk centrality comes up as a good candidate  for capturing the
notion of centrality on semantic networks.

%We hope we presented enough evidence that this topic  deserves
%to be explored and researched. In this paper we stated the problem at
%the formal level, show it is hard almost on any reasonable point that has been
%proposed, show that one can make advances and the very notion
%makes sense, and study one measure that, besides not have been
%explored before, seems to be a good candidate among existing
%tools to speak about ``similarity'' and ``centrality'' among semantic
%notions in a network.

 The above opens at least two lines of research on this topic: First,
to find good approximation algorithms to compute group centrality
particularly for random walk centrality. Second, the task of building
a good benchmark of big semantic networks together with their
candidate sets to match the most important concepts in the network.

%%%%%%%%%%%%%%%%%%%%%%%%%%%%%%%
%-------------------------------------------------------------------------------------------------------------------------------------

\newpage
\appendix

\section{Appendix}

%Proximity Semantic Network 1:
\begin{figure}[!ht]
  \centering
    \includegraphics[width=0.7\textwidth]{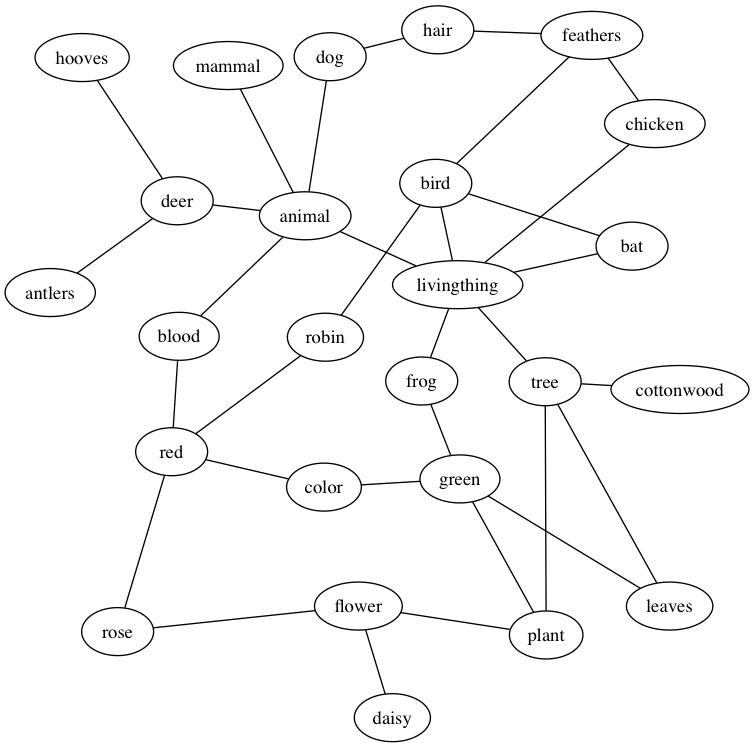}
  \caption{Proximity Semantic Network 1.}
\end{figure}

Proximity Semantic Network 1 Labels:
\begin{multicols}{2}
\begin{itemize}
    \item 0.- animal
    \item 1.- antlers
    \item 2.- bat
    \item 3.- bird
    \item 4.- blood
    \item 5.- chicken
    \item 6.- color
    \item 7.- cottonwood
    \item 8.- daisy
    \item 9.- deer
    \item 10.- dog
    \item 11.- feathers
    \item 12.- flower
    \item 13.- frog
    \item 14.- green
    \item 15.- hair
    \item 16.- hooves
    \item 17.- leaves
    \item 18.- livingthing
    \item 19.- mammal
    \item 20.- plant
    \item 21.- red
    \item 22.- robin
    \item 23.- rose
    \item 24.- tree
\end{itemize}
\end{multicols}

\begin{table}[h!]
        \centering
        {
        \scriptsize
        \begin{tabularx}{\textwidth}{| c | Y | Y | Y || Y |}
            \hline
            ~n & degree & closeness & betweenness & random-walk \\ \hline
            1 & \{11\} & \{11\} & \{11\} & \{11\} \\ \hline
            2 & \{4, 11\} & \{2, 13\} , \{4, 11\} & \{4, 11\} & \{4, 11\} \\ \hline
            3 & \{4, 11, 13\} , \{0, 11, 13\} , ... (9) & \{0, 11, 13\} , \{4, 11, 13\} , ... (4) & \{4, 11, 13\} & \{0, 11, 13\} \\ \hline
            4 & \{2, 4, 11, 13\} , ... (23) & \{2, 4, 11, 13\} , ... (5)] & \{2, 4, 11, 13\} & \{2, 4, 11, 13\} \\ \hline
        \end{tabularx}
        }
        \smallskip
        \caption{
            Solutions for group centrality on Proximity Semantic Network 1.
        }
        \label{table:proximity}
    \end{table}

%Proximity Semantic Network 2:
\begin{figure}[!ht]
  \centering
    \includegraphics[width=0.7\textwidth]{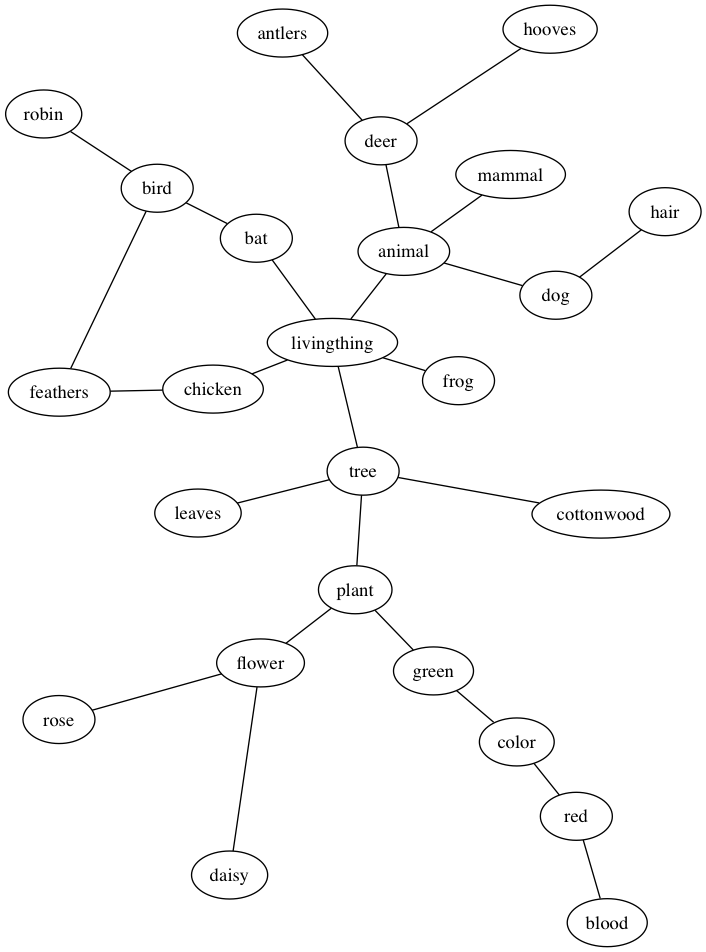}
  \caption{Proximity Semantic Network 2.}
\end{figure}
\vspace{2cm}
Proximity Semantic Network 2 Labels:
\begin{multicols}{2}
\begin{itemize}
    \item 0.- animal
    \item 1.- antlers
    \item 2.- bat
    \item 3.- bird
    \item 4.- blood
    \item 5.- chicken
    \item 6.- color
    \item 7.- cottonwood
    \item 8.- daisy
    \item 9.- deer
    \item 10.- dog
    \item 11.- feathers
    \item 12.- flower
    \item 13.- frog
    \item 14.- green
    \item 15.- hair
    \item 16.- hooves
    \item 17.- leaves
    \item 18.- livingthing
    \item 19.- mammal
    \item 20.- plant
    \item 21.- red
    \item 22.- robin
    \item 23.- rose
    \item 24.- tree

\end{itemize}
\end{multicols}

    \begin{table}[h!]
        \centering
        {
        \scriptsize
        \begin{tabularx}{\textwidth}{| c | Y | Y | Y || Y |}
            \hline
~n~ & degree & closeness & betweenness & random-walk \\ \hline
1 & \{11\} & \{11\} & \{11\} & \{11\} \\ \hline
2 & \{0, 11\} & \{2, 11\} , ... (2) & \{2, 11\} & \{2, 11\} \\ \hline
3 & \{2, 13, 21\} , ... (8) & \{2, 13, 21\} & \{2, 11, 13\} & \{2, 13, 21\} \\ \hline
4 & \{2, 4, 13, 21\} , ... (16) & \{2, 4, 13, 21\} & \{2, 11, 13, 21\} & \{2, 4, 13, 21\} \\ \hline
        \end{tabularx}
        }
        \smallskip
        \caption{
            Solutions for group centrality on Proximity Semantic Network 2.
        }
        \label{table:proximity11}
    \end{table}

%Proximity network built by non-specialists:

\begin{figure}[!ht]
  \centering
    \includegraphics[width=0.7\textwidth]{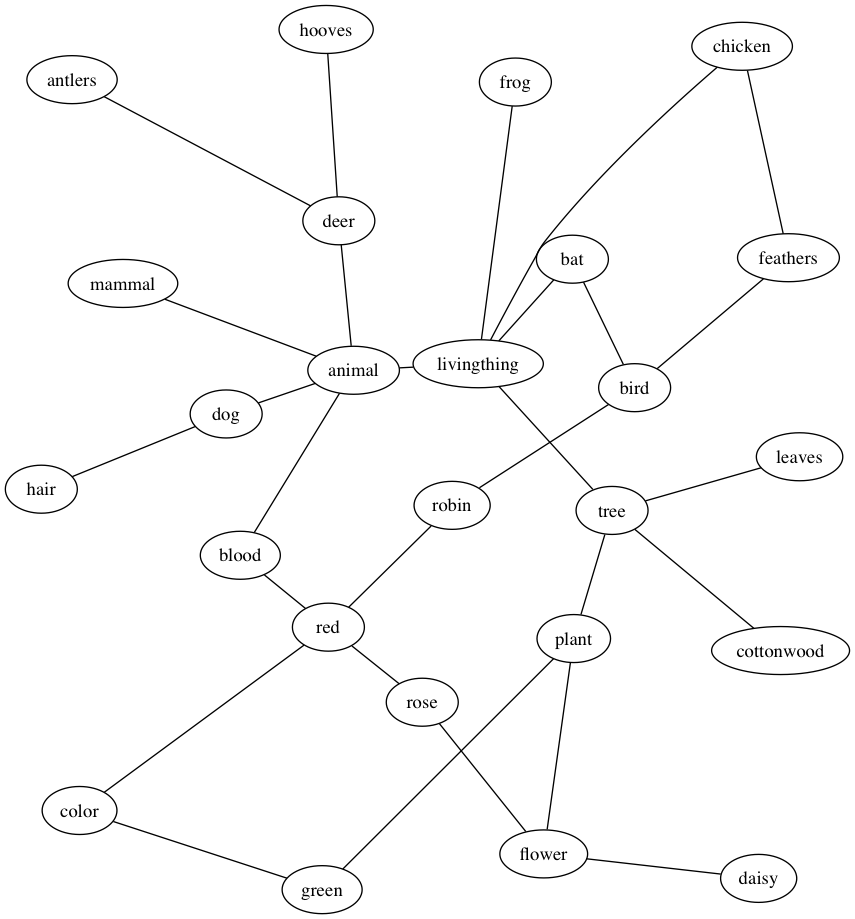}
  \caption{Proximity network built by non-specialists.}
\end{figure}

\newpage
Proximity network built by non-specialists labels:
\begin{multicols}{2}
\begin{itemize}
    \item 0.- animal
    \item 1.- antlers
    \item 2.- bat
    \item 3.- bird
    \item 4.- blood
    \item 5.- chicken
    \item 6.- color
    \item 7.- cottonwood
    \item 8.- daisy
    \item 9.- deer
    \item 10.- dog
    \item 11.- feathers
    \item 12.- flower
    \item 13.- frog
    \item 14.- green
    \item 15.- hair
    \item 16.- hooves
    \item 17.- leaves
    \item 18.- livingthing
    \item 19.- mammal
    \item 20.- plant
    \item 21.- red
    \item 22.- robin
    \item 23.- rose
    \item 24.- tree
\end{itemize}

\end{multicols}

    \begin{table}[h!]
        \centering
        {
        \scriptsize
        \begin{tabularx}{\textwidth}{| c | Y | Y | Y || Y |}
            \hline
~n~ & degree & closeness & betweenness & random-walk \\ \hline
1 & \{6\} , \{17\} & \{6\} & \{17\} & \{6\} \\ \hline
2 & \{0, 6\} & \{0, 6\} , \{10, 17\} & \{0, 6\} & \{0, 6\} \\ \hline
3 & \{3, 13, 17\} , \{3, 10, 17\}, ... (10) & \{3, 10, 17\} & \{0, 13, 17\} & \{3, 10, 17\} \\ \hline
4 & \{0, 6, 10, 17\} , ... (18) & \{0, 6, 10, 17\} , ... (6) & \{0, 6, 10, 17\} & \{0, 6, 10, 17\} \\ \hline
        \end{tabularx}
        }
        \smallskip
        \caption{
          Solutions for group centrality on Proximity Semantic Network constructed by non-biologist students.
        }
        \label{table:proximity2}
    \end{table}

\begin{figure}[!ht]
  \centering
    \includegraphics[width=0.5\textwidth]{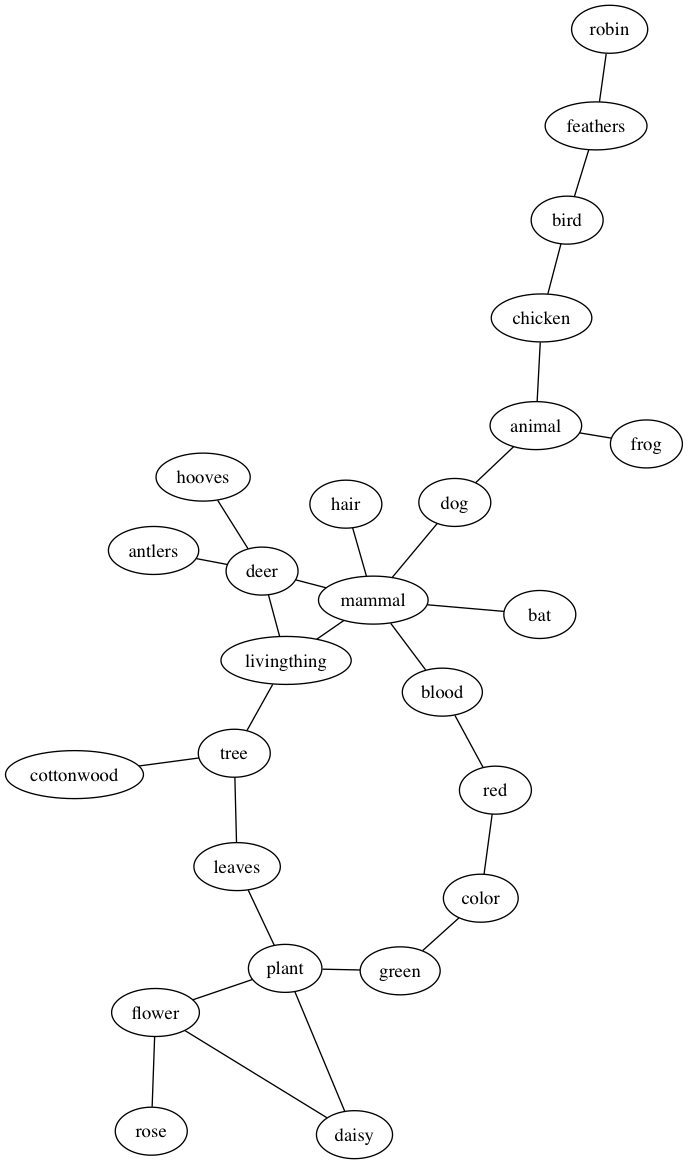}
  \caption{Proximity network built by biologists.}
\end{figure}
\newpage
Proximity network built by biologists labels:
\begin{multicols}{2}
\begin{itemize}
    \item 0.- animal
    \item 1.- antlers
    \item 2.- bat
    \item 3.- bird
    \item 4.- blood
    \item 5.- chicken
    \item 6.- color
    \item 7.- cottonwood
    \item 8.- daisy
    \item 9.- deer
    \item 10.- dog
    \item 11.- feathers
    \item 12.- flower
    \item 13.- frog
    \item 14.- green
    \item 15.- hair
    \item 16.- hooves
    \item 17.- leaves
    \item 18.- livingthing
    \item 19.- mammal
    \item 20.- plant
    \item 21.- red
    \item 22.- robin
    \item 23.- rose
    \item 24.- tree
\end{itemize}
\end{multicols}

    \begin{table}[h!]
        \centering
        {
        \scriptsize
        \begin{tabular}{| c | c | c | c || c |}
            \hline
            ~n~ & degree & closeness & betweenness & random-walk \\ \hline
1 & \{12\} & \{12\} & \{12\} & \{12\} \\ \hline
2 & \{3, 12\} & \{3, 12\} & \{3, 12\} & \{3, 12\} \\ \hline
3 & \{3, 12, 19\}, \{3, 12, 21\} & \{3, 12, 21\}, \{3, 12, 22\} & \{3, 12, 19\}, \{3, 12, 21\} & \{3, 12, 22\} \\
 & \{3, 12, 22\}, \{3, 12, 23\} & \{3, 12, 23\}  &  &  \\ \hline
4 & \{3, 12, 19, 23\} & \{3, 12, 19, 23\} & \{3, 12, 15, 19\}, \{3, 12, 15, 21\} & \{3, 5, 12, 22\} \\ \hline

        \end{tabular}
        }
        \smallskip
        \caption{
          Solutions for group centrality on Proximity Semantic Network constructed by biologist students.
        }
        \label{table:proximity3}
    \end{table}

%\newpage
%Sample 1 from Perception network:

\begin{figure}[!ht]
  \centering
    \includegraphics[width=0.7\textwidth]{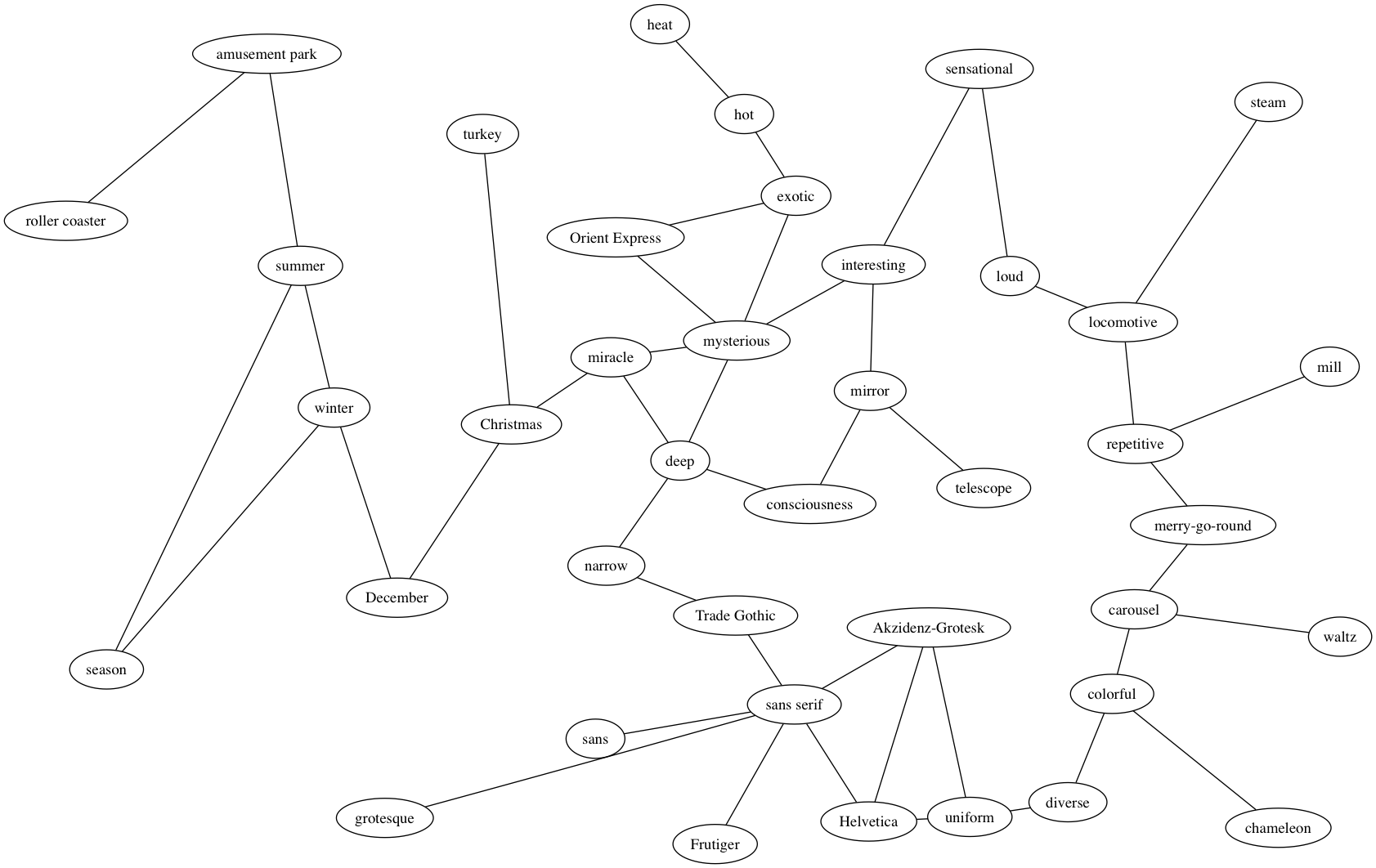}
  \caption{Sample 1 from Perception Network.}
\end{figure}

\newpage
Sample 1 from Perception network labels:

\begin{multicols}{2}
\begin{itemize}
    \item 0.- Akzidenz-Grotesk
    \item 1.- Christmas
    \item 2.- December
    \item 3.- Frutiger
    \item 4.- Helvetica
    \item 5.- Orient Express
    \item 6.- Trade Gothic
    \item 7.- amusement park
    \item 8.- carousel
    \item 9.- chameleon
    \item 10.- colorful
    \item 11.- consciousness
    \item 12.- deep
    \item 13.- diverse
    \item 14.- exotic
    \item 15.- grotesque
    \item 16.- heat
    \item 17.- hot
    \item 18.- interesting
    \item 19.- locomotive
    \item 20.- loud
    \item 21.- merry-go-round
    \item 22.- mill
    \item 23.- miracle
    \item 24.- mirror
    \item 25.- mysterious
    \item 26.- narrow
    \item 27.- repetitive
    \item 28.- roller coaster
    \item 29.- sans
    \item 30.- sans serif
    \item 31.- season
    \item 32.- sensational
    \item 33.- steam
    \item 34.- summer
    \item 35.- telescope
    \item 36.- turkey
    \item 37.- uniform
    \item 38.- waltz
    \item 39.- winter
\end{itemize}
\end{multicols}

\begin{table}[h!]
        \centering
        {
        \scriptsize
        \begin{tabular}{| c | c | c | c || c |}
            \hline
~n~ & degree & closeness & betweenness & random-walk \\ \hline
1 & \{25\} & \{14\} & \{14\} & \{13\} \\ \hline
2 & \{14, 25\} & \{5, 14\} , \{7, 13\} , \{7, 14\} , \{8, 13\} & \{14, 25\} & \{6, 14\} \\ \hline
3 & \{4, 14, 25\} , ... (5) & \{4, 14, 25\} , ... (4) & \{4, 14, 25\} & \{0, 6, 14\} \\ \hline
4 & \{0, 4, 14, 25\} , ... (10) & \{1, 4, 14, 25\} , ... (2) & \{4, 12, 14, 25\} & \{0, 4, 14, 25\} \\ \hline
        \end{tabular}
        }
        \smallskip
        \caption{
            Solutions for group centrality on Perception sample 1.
        }
        \label{table:perception1}
    \end{table}

%\newpage
%Sample 2 from Perception network:

\begin{figure}[!ht]
  \centering
    \includegraphics[width=0.7\textwidth]{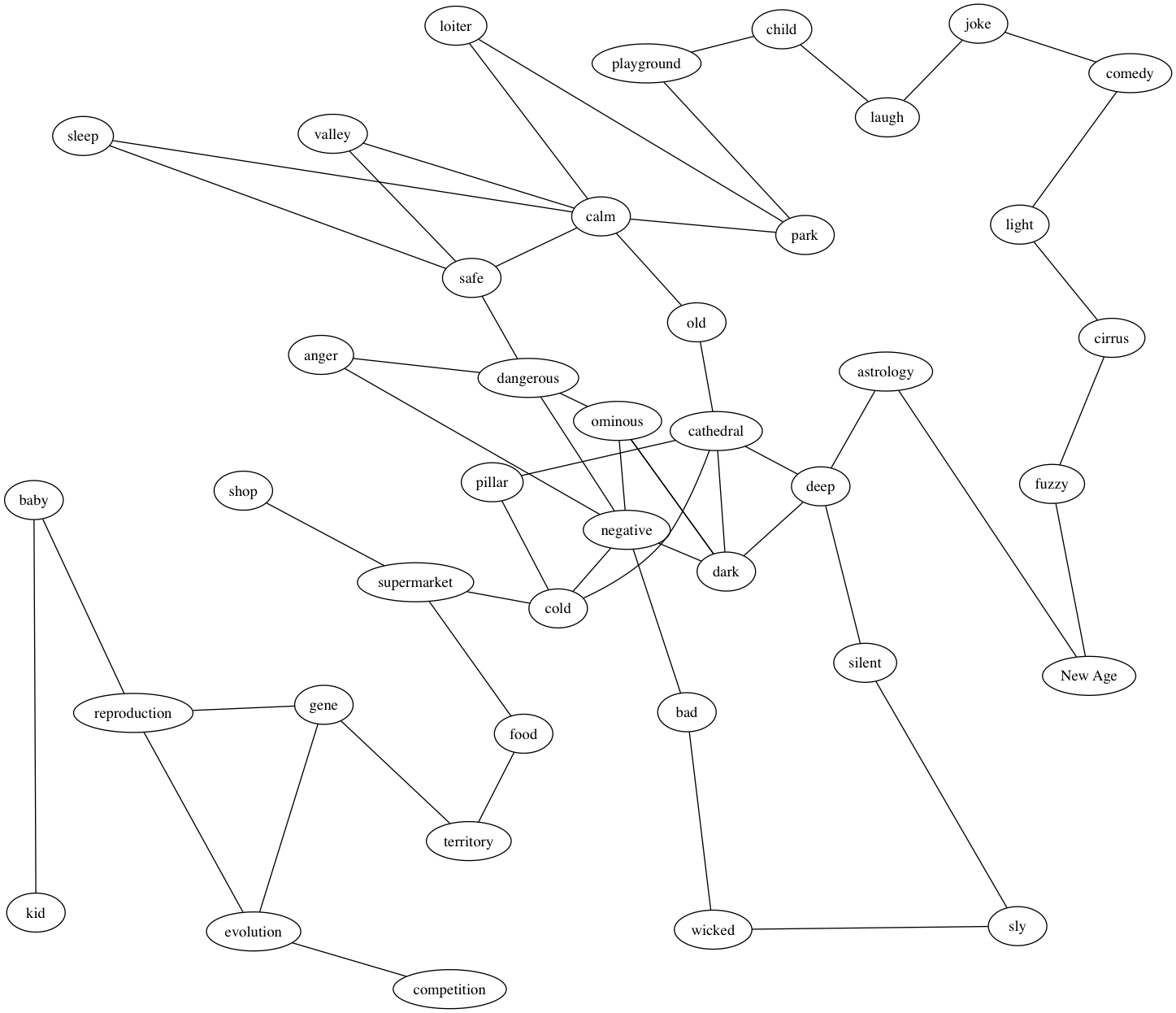}
  \caption{Sample 2 from Perception Network.}
\end{figure}

Sample 2 from Perception network labels:
\begin{multicols}{2}
\begin{itemize}
\item 0.- New Age
\item 1.- anger
\item 2.- astrology
\item 3.- baby
\item 4.- bad
\item 5.- calm
\item 6.- cathedral
\item 7.- child
\item 8.- cirrus
\item 9.- cold
\item 10.- comedy
\item 11.- competition
\item 12.- dangerous
\item 13.- dark
\item 14.- deep
\item 15.- evolution
\item 16.- food
\item 17.- fuzzy
\item 18.- gene
\item 19.- joke
\item 20.- kid
\item 21.- laugh
\item 22.- light
\item 23.- loiter
\item 24.- negative
\item 25.- old
\item 26.- ominous
\item 27.- park
\item 28.- pillar
\item 29.- playground
\item 30.- reproduction
\item 31.- safe
\item 32.- shop
\item 33.- silent
\item 34.- sleep
\item 35.- sly
\item 36.- supermarket
\item 37.- territory
\item 38.- valley
\item 39.- wicked
\end{itemize}
\end{multicols}

    \begin{table}[h!]
        \centering
        {
        \scriptsize
        \begin{tabular}{| c | c | c | c || c |}
            \hline
~n~ & degree & closeness & betweenness & random-walk \\ \hline
1 & \{29\} , ... (1) & \{15\} & \{1\} & \{29\} \\ \hline
2 & \{2, 29\} & \{0, 15\} , ... (1) & \{1, 2\} & \{29, 39\} \\ \hline
3 & \{2, 8, 29\} , ... (3) & \{0, 15, 32\} , ... (7) & \{1, 2, 8\} & \{2, 8, 39\} \\ \hline
4 & \{2, 8, 29, 39\} , ... (2) & \{2, 26, 29, 39\} , ... (1) & \{2, 8, 10, 29\} & \{2, 26, 29, 39\} \\ \hline
        \end{tabular}
        }
        \smallskip
        \caption{
            Solutions for group centrality on Perception sample 2.
        }
        \label{table:perception2}
    \end{table}

\begin{figure}[!ht]
  \centering
    \includegraphics[width=0.7\textwidth]{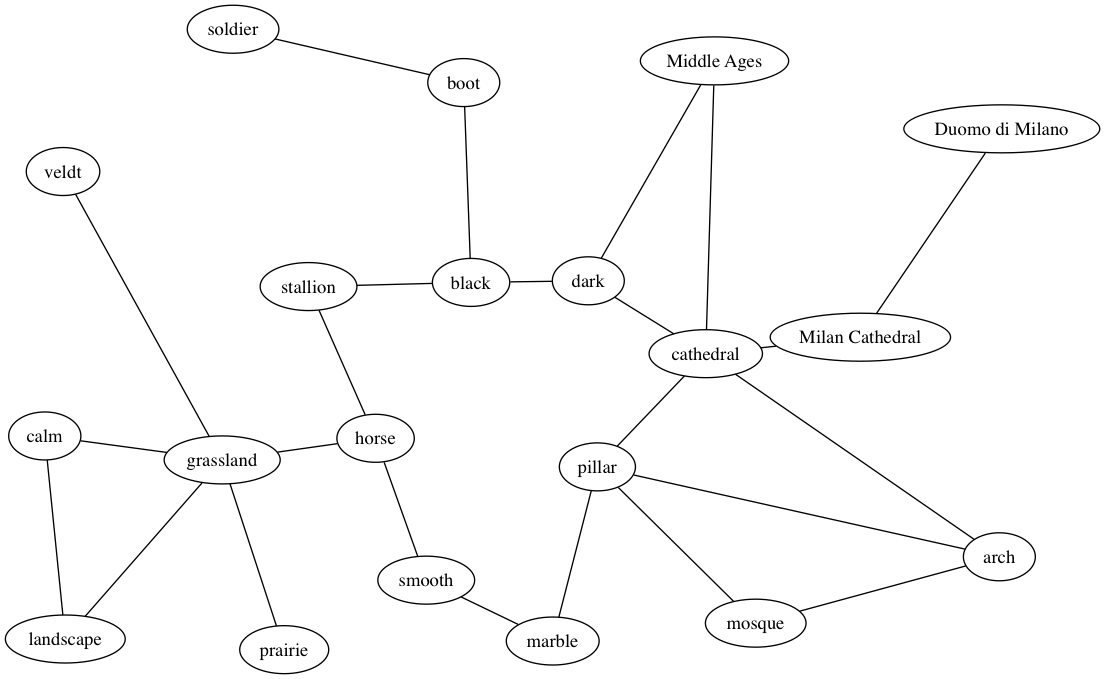}
  \caption{Sample 3 from Perception Network.}
\end{figure}

Sample 3 from Perception network labels:

\begin{multicols}{2}
\begin{itemize}
    \item 0.- Duomo di Milano
    \item 1.- Middle Ages
    \item 2.- Milan Cathedral
    \item 3.- arch
    \item 4.- black
    \item 5.- boot
    \item 6.- calm
    \item 7.- cathedral
    \item 8.- dark
    \item 9.- grassland
    \item 10.- horse
    \item 11.- landscape
    \item 12.- marble
    \item 13.- mosque
    \item 14.- pillar
    \item 15.- prairie
    \item 16.- smooth
    \item 17.- soldier
    \item 18.- stallion
    \item 19.- veldt
\end{itemize}

\end{multicols}

     \begin{table}[h]
         \centering
         {
         \scriptsize
         \begin{tabular}{| c | c | c | c || c |}
             \hline
 ~n~ & degree & closeness & betweenness & random-walk \\ \hline
 1 & \{11\} & \{11\} & \{11\} & \{11\} \\ \hline
 2 & \{4, 11\} & \{2, 13\} , \{4, 11\} & \{4, 11\} & \{4, 11\} \\ \hline
 3 & \{0, 11, 13\} , \{4, 11, 13\} , ... (9) & \{0, 11, 13\} , \{4, 11, 13\} , ... (4) & \{4, 11, 13\} & \{0, 11, 13\} \\ \hline
 4 & \{2, 4, 11, 13\} , ... (23) & \{2, 4, 11, 13\} , ... (5) & \{2, 4, 11, 13\} & \{2, 4, 11, 13\} \\ \hline
         \end{tabular}
         }
         \smallskip
         \caption{
             Solutions for group centrality on Perception sample 3.
         }
         \label{table:perception3}
     \end{table}

%\newpage
%Sample 1 from DBpedia Categories network:
\newpage
\begin{figure}[!ht]
  \centering
    \includegraphics[width=0.7\textwidth]{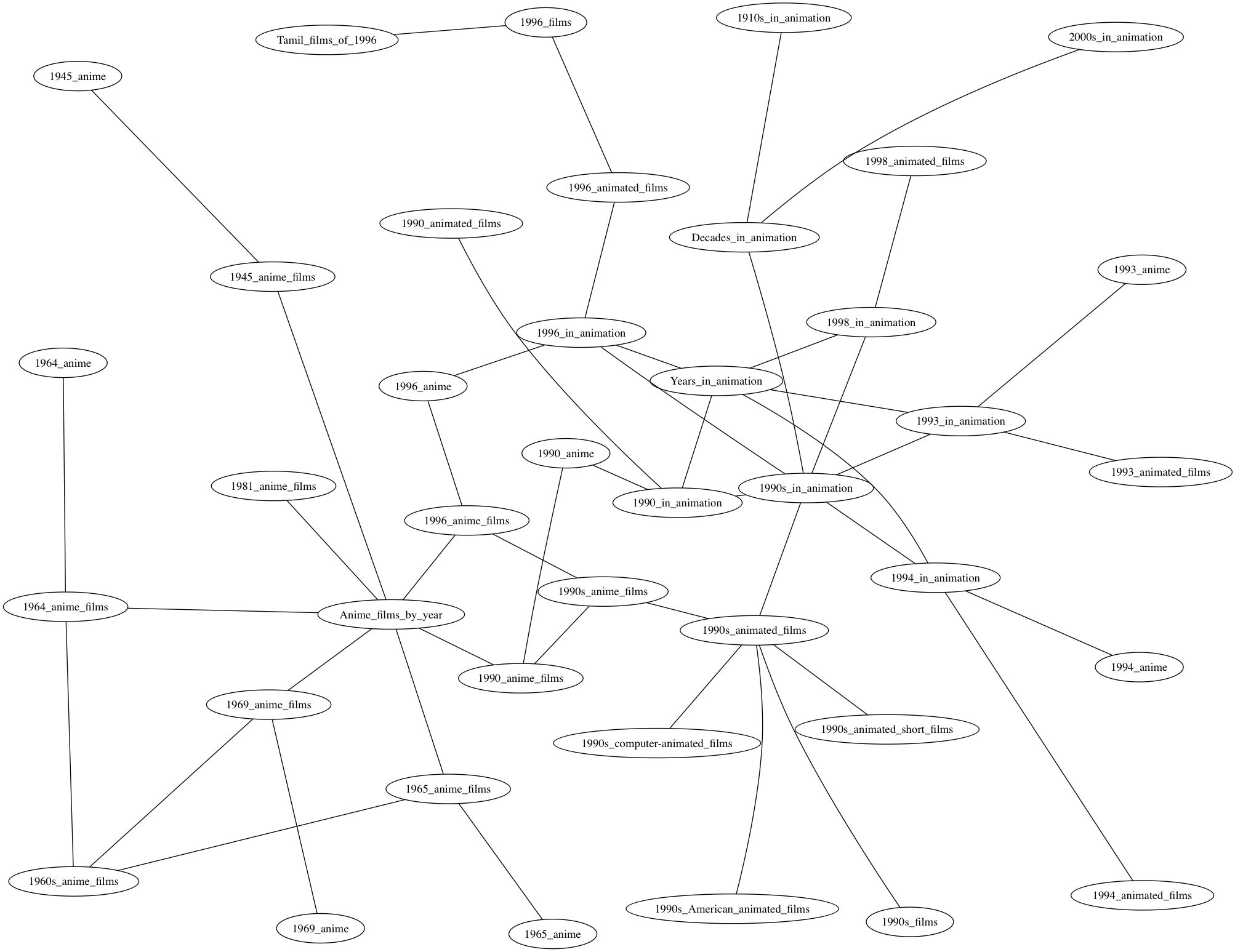}
  \caption{Sample 1 from DBpedia Categories network.}
\end{figure}
\vspace{2cm}
Sample 1 from DBpedia Categories network labels:

\begin{multicols}{2}
\begin{itemize}
    \item 0.- 1996\_films
    \item 1.- 1990s\_films
    \item 2.- 2000s\_in\_animation
    \item 3.- Decades\_in\_animation
    \item 4.- 1990s\_in\_animation
    \item 5.- 1910s\_in\_animation
    \item 6.- 1990s\_animated\_films
    \item 7.- Years\_in\_animation
    \item 8.- 1964\_anime
    \item 9.- 1965\_anime
    \item 10.- 1969\_anime
    \item 11.- 1990\_anime
    \item 12.- 1990\_in\_animation
    \item 13.- 1993\_anime
    \item 14.- 1993\_in\_animation
    \item 15.- 1994\_anime
    \item 16.- 1994\_in\_animation
    \item 17.- 1996\_anime
    \item 18.- 1996\_in\_animation
    \item 19.- 1998\_in\_animation
    \item 20.- Anime\_films\_by\_year
    \item 21.- 1990\_anime\_films
    \item 22.- 1990s\_anime\_films
    \item 23.- 1996\_anime\_films
    \item 24.- 1981\_anime\_films
    \item 25.- 1960s\_anime\_films
    \item 26.- 1998\_animated\_films
    \item 27.- 1996\_animated\_films
    \item 28.- 1994\_animated\_films
    \item 29.- 1993\_animated\_films
    \item 30.- 1990\_animated\_films
    \item 31.- 1965\_anime\_films
    \item 32.- 1969\_anime\_films
    \item 33.- 1945\_anime\_films
    \item 34.- 1945\_anime
    \item 35.- 1990s\_American\_animated\_films
    \item 36.- 1990s\_animated\_short\_films
    \item 37.- Tamil\_films\_of\_1996
    \item 38.- 1964\_anime\_films
    \item 39.- 1990s\_computer-animated\_films
\end{itemize}

\end{multicols}

    \begin{table}[h]
        \centering
        {
        \scriptsize
        \begin{tabular}{| c | c | c | c || c |}
            \hline
~n~ & degree & closeness & betweenness & random-walk \\ \hline
1 & \{4\} , \{20\} & \{4\} & \{20\} & \{4\} \\ \hline
2 & \{4, 20\} & \{4, 20\} & \{4, 20\} & \{4, 20\} \\ \hline
3 & \{4, 6, 20\} , \{6, 7, 20\} & \{0, 4, 20\} & \{4, 7, 20\} & \{4, 18, 20\} \\ \hline
4 & \{4, 6, 18, 20\} , \{0, 4, 6, 20\} & \{0, 4, 6, 20\} & \{4, 6, 7, 20\} & \{4, 6, 18, 20\} \\ \hline
        \end{tabular}
        }
        \smallskip
        \caption{
            Solutions for group centrality on DBpedia sample 1.
        }
        \label{table:dbpedia1}
    \end{table}

%\newpage
%Sample 2 from DBpedia Categories network:

\begin{figure}[!ht]
  \centering
    \includegraphics[width=0.7\textwidth]{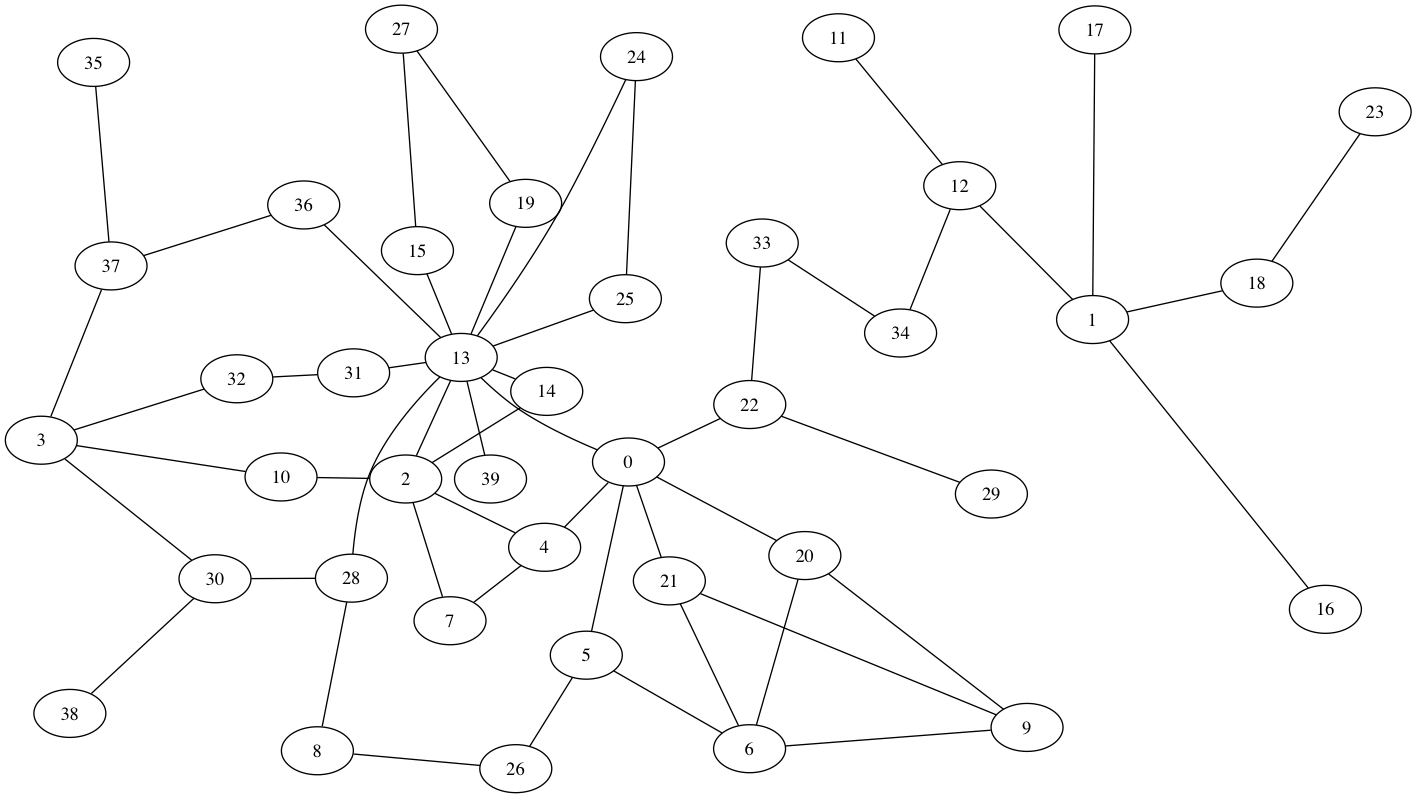}
  \caption{Sample 2 from DBpedia Categories network.}
\end{figure}
\newpage
Sample 2 from DBpedia Categories network labels:
\begin{itemize}
    \item 0.- Wikipedia\_categories\_named\_after\_populated\_places\_in\_Italy
    \item 1.- Provinces\_of\_Italy
    \item 2.- Cities\_and\_towns\_in\_Emilia-Romagna
    \item 3.- Emilia-Romagna
    \item 4.- Parma
    \item 5.- Turin
    \item 6.- Cities\_and\_towns\_in\_Piedmont
    \item 7.- Communes\_of\_the\_Province\_of\_Parma
    \item 8.- Education\_in\_Italy\_by\_city
    \item 9.- Communes\_of\_the\_Province\_of\_Alessandria
    \item 10.- Geography\_of\_Emilia-Romagna
    \item 11.- Province\_of\_Pavia\_geography\_stubs
    \item 12.- Province\_of\_Pavia
    \item 13.- Bologna
    \item 14.- Communes\_of\_the\_Province\_of\_Bologna
    \item 15.- Buildings\_and\_structures\_in\_Bologna
    \item 16.- Province\_of\_Arezzo
    \item 17.- Province\_of\_Ancona
    \item 18.- Province\_of\_Lucca
    \item 19.- Visitor\_attractions\_in\_Bologna
    \item 20.- Tortona
    \item 21.- Novi\_Ligure
    \item 22.- Pavia
    \item 23.- Communes\_of\_the\_Province\_of\_Lucca
    \item 24.- Roman\_Catholic\_archbishops\_of\_Bologna
    \item 25.- History\_of\_Bologna
    \item 26.- Education\_in\_Turin
    \item 27.- Palaces\_in\_Bologna
    \item 28.- Education\_in\_Bologna
    \item 29.- University\_of\_Pavia
    \item 30.- Education\_in\_Emilia-Romagna
    \item 31.- Culture\_in\_Bologna
    \item 32.- Culture\_in\_Emilia-Romagna
    \item 33.- Buildings\_and\_structures\_in\_Pavia
    \item 34.- Buildings\_and\_structures\_in\_the\_Province\_of\_Pavia
    \item 35.- Companies\_by\_region\_of\_Italy
    \item 36.- Companies\_based\_in\_Bologna
    \item 37.- Companies\_based\_in\_Emilia-Romagna
    \item 38.- Education\_in\_Italy\_by\_region
    \item 39.- Media\_in\_Bologna
\end{itemize}

    \begin{table}[h!]
        \centering
        {
        \scriptsize
        \begin{tabular}{| c | c | c | c || c |}
            \hline
~n~ & degree & closeness & betweenness & random-walk \\ \hline
1 & \{13\} & \{13\} & \{13\} & \{13\} \\ \hline
2 & \{0, 13\} , \{1, 13\} , ... (2) & \{1, 13\} & \{0, 13\} & \{12, 13\} \\ \hline
3 & \{0, 1, 13\} , ... (4) & \{0, 1, 13\} & \{0, 1, 13\} & \{0, 1, 13\} \\ \hline
4 & \{0, 1, 3, 13\} , ... (1) & \{0, 1, 3, 13\} & \{0, 1, 3, 13\} & \{0, 1, 3, 13\} \\ \hline
        \end{tabular}
        }
        \smallskip
        \caption{
            Solutions for group centrality on DBpedia sample 2.
        }
        \label{table:dbpedia2}
    \end{table}

%\newpage
%Sample 3 from DBpedia Categories network:

\begin{figure}[!ht]
  \centering
    \includegraphics[width=0.7\textwidth]{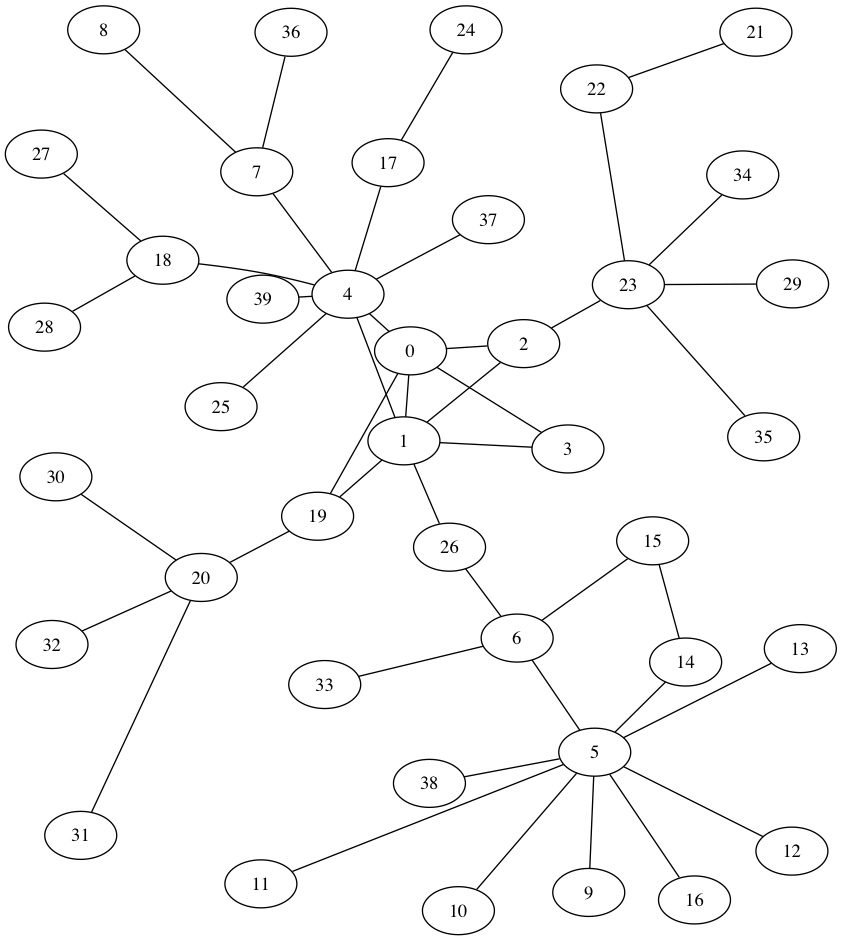}
  \caption{Sample 3 from DBpedia Categories network.}
\end{figure}
\newpage
Sample 3 from DBpedia Categories network labels:

\begin{multicols}{2}
\begin{itemize}
  \item 0.- Neognathae
    \item 1.- Birds\_by\_classification
    \item 2.- Procellariiformes
    \item 3.- Phoenicopteriformes
    \item 4.- Pelecaniformes
    \item 5.- Hornbills
    \item 6.- Bucerotidae
    \item 7.- Ardeidae
    \item 8.- Egretta
    \item 9.- Aceros
    \item 10.- Anthracoceros
    \item 11.- Buceros
    \item 12.- Ceratogymna
    \item 13.- Ocyceros
    \item 14.- Penelopides
    \item 15.- Bucerotinae
    \item 16.- Tockus
    \item 17.- Phalacrocoracidae
    \item 18.- Fregatidae
    \item 19.- Falconiformes
    \item 20.- Accipitridae
    \item 21.- Oceanodroma
    \item 22.- Hydrobatinae
    \item 23.- Hydrobatidae
    \item 24.- Phalacrocorax
    \item 25.- Plotopteridae
    \item 26.- Bucerotiformes
    \item 27.- Limnofregata
    \item 28.- Fregata
    \item 29.- Oceanitinae
    \item 30.- Aviceda
    \item 31.- Dryotriorchis
    \item 32.- Erythrotriorchis
    \item 33.- Anorrhinus
    \item 34.- Fregetta
    \item 35.- Oceanites
    \item 36.- Botaurus
    \item 37.- Anhingidae
    \item 38.- Rhyticeros
    \item 39.- Pelecanus
\end{itemize}
\end{multicols}

     \begin{table}[h]
         \centering
         {
         \scriptsize
         \begin{tabular}{| c | c | c | c || c |}
             \hline
 ~n~ & degree & closeness & betweenness & random-walk \\ \hline
 1 & \{5\} & \{1\} & \{1\} & \{1\} \\ \hline
 2 & \{4, 5\} & \{1, 5\} & \{4, 5\} & \{1, 5\} \\ \hline
 3 & \{4, 5, 23\} & \{4, 5, 23\} & \{4, 5, 23\} & \{4, 5, 23\} \\ \hline
 4 & \{4, 5, 20, 23\} & \{4, 5, 20, 23\} & \{4, 5, 20, 23\} & \{4, 5, 20, 23\} \\ \hline
         \end{tabular}
         }
         \smallskip
         \caption{
             Solutions for group centrality on DBpedia sample 3.
         }
         \label{table:dbpedia3}
     \end{table}

%\newpage
%Sample 1 from Roget's Thesaurus network:

\begin{figure}[!ht]
  \centering
    \includegraphics[width=0.7\textwidth]{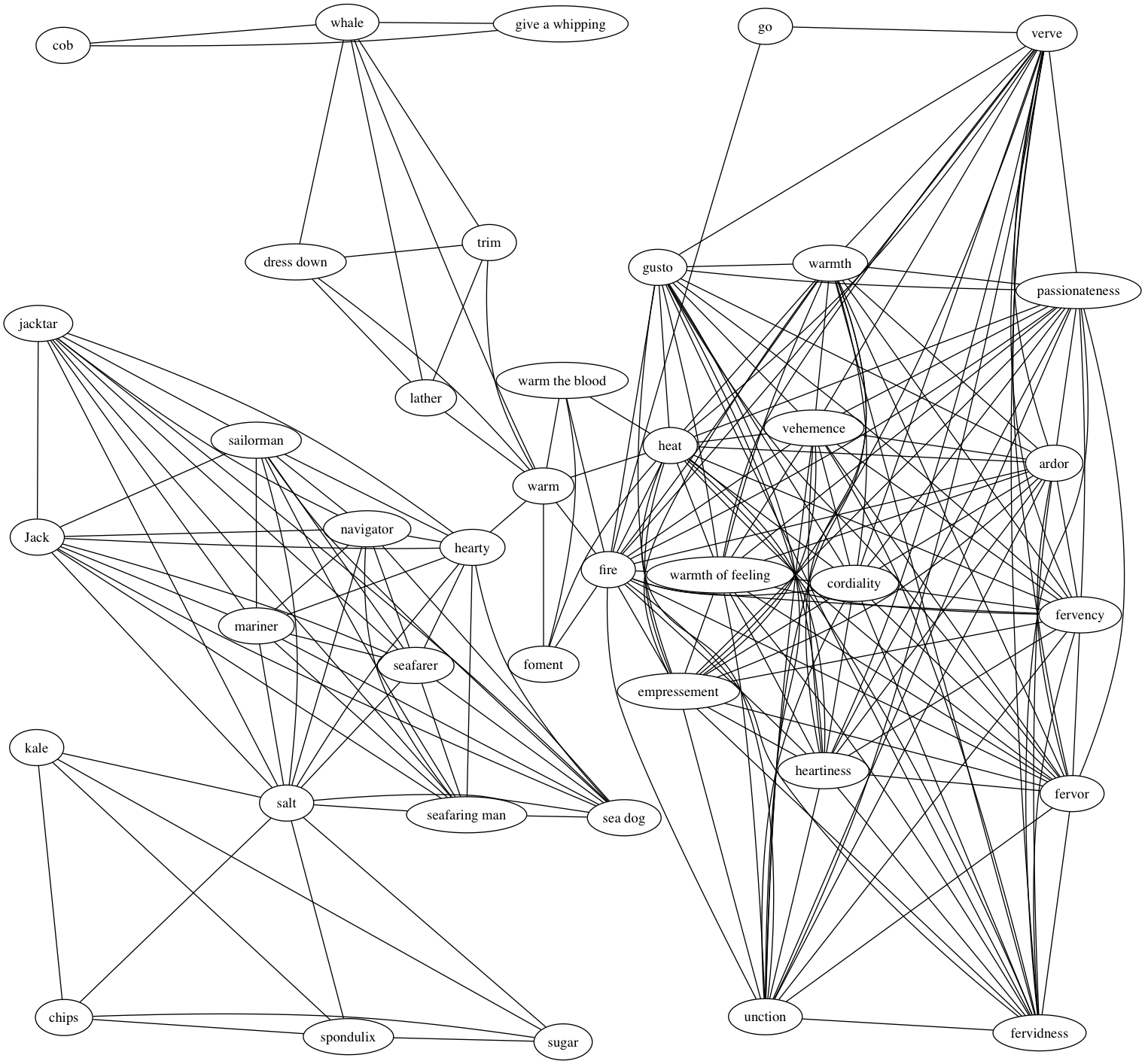}
  \caption{Sample 1 from Roget's Thesaurus network.}
\end{figure}

\newpage
Sample 1 from Roget's Thesaurus network labels:
\begin{multicols}{2}
\begin{itemize}
    \item 0.- Jack
    \item 1.- ardor
    \item 2.- chips
    \item 3.- cob
    \item 4.- cordiality
    \item 5.- dress down
    \item 6.- empressement
    \item 7.- fervency
    \item 8.- fervidness
    \item 9.- fervor
    \item 10.- fire
    \item 11.- foment
    \item 12.- give a whipping
    \item 13.- go
    \item 14.- gusto
    \item 15.- heartiness
    \item 16.- hearty
    \item 17.- heat
    \item 18.- jacktar
    \item 19.- kale
    \item 20.- lather
    \item 21.- mariner
    \item 22.- navigator
    \item 23.- passionateness
    \item 24.- sailorman
    \item 25.- salt
    \item 26.- sea dog
    \item 27.- seafarer
    \item 28.- seafaring man
    \item 29.- spondulix
    \item 30.- sugar
    \item 31.- trim
    \item 32.- unction
    \item 33.- vehemence
    \item 34.- verve
    \item 35.- warm
    \item 36.- warm the blood
    \item 37.- warmth
    \item 38.- warmth of feeling
    \item 39.- whale
\end{itemize}

\end{multicols}

    \begin{table}[h!]
        \centering
        {
        \scriptsize
        \begin{tabular}{| c | c | c | c || c |}
            \hline
~n~ & degree & closeness & betweenness & random-walk \\ \hline
1 & \{10\} & \{35\} & \{35\} & \{10\} \\ \hline
2 & \{10, 25\} & \{10, 25\} & \{25, 35\} & \{10, 25\} \\ \hline
3 & \{10, 25, 39\} & \{10, 25, 39\} & \{10, 25, 35\} & \{10, 25, 39\} \\ \hline
4 & \{10, 17, 25, 39\} , ... (87) & \{10, 17, 25, 39\} , ... (87) & \{10, 17, 25, 35\} & \{10, 17, 25, 39\} \\ \hline
        \end{tabular}
        }
        \smallskip
        \caption{
            Solutions for group centrality on Roget's Thesaurus sample 1.
        }
        \label{table:thesaurus1}
    \end{table}

%\newpage
%Sample 2 from Roget's Thesaurus network:
\newpage
\begin{figure}[!ht]
  \centering
    \includegraphics[width=0.7\textwidth]{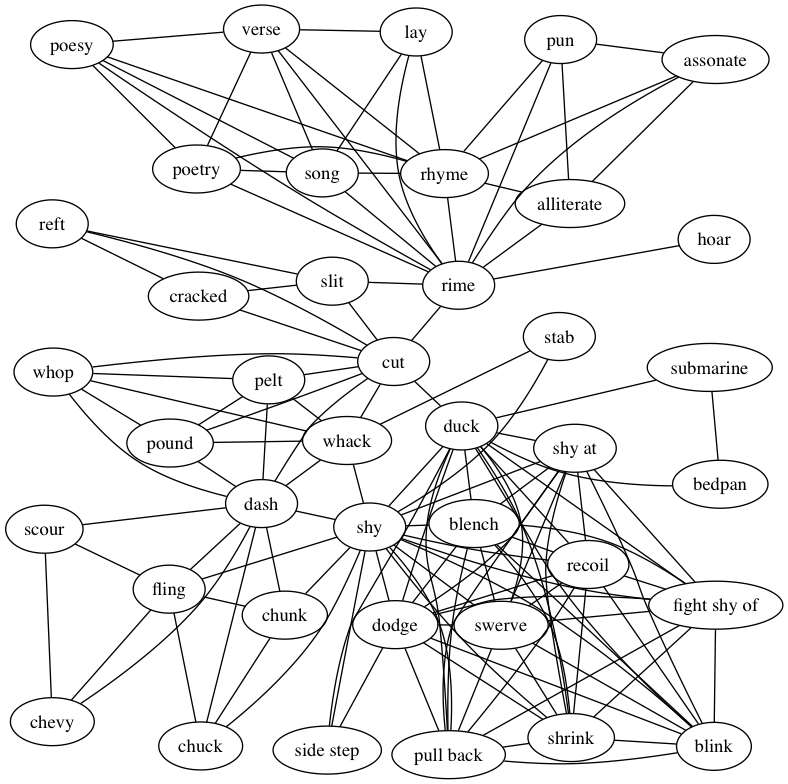}
  \caption{Sample 2 from Roget's Thesaurus network.}
\end{figure}

Sample 2 from Roget's Thesaurus network labels:

\begin{multicols}{2}
\begin{itemize}
    \item 0.- alliterate
    \item 1.- assonate
    \item 2.- bedpan
    \item 3.- blench
    \item 4.- blink
    \item 5.- chevy
    \item 6.- chuck
    \item 7.- chunk
    \item 8.- cracked
    \item 9.- cut
    \item 10.- dash
    \item 11.- dodge
    \item 12.- duck
    \item 13.- fight shy of
    \item 14.- fling
    \item 15.- hoar
    \item 16.- lay
    \item 17.- pelt
    \item 18.- poesy
    \item 19.- poetry
    \item 20.- pound
    \item 21.- pull back
    \item 22.- pun
    \item 23.- recoil
    \item 24.- reft
    \item 25.- rhyme
    \item 26.- rime
    \item 27.- scour
    \item 28.- shrink
    \item 29.- shy
    \item 30.- shy at
    \item 31.- side step
    \item 32.- slit
    \item 33.- song
    \item 34.- stab
    \item 35.- submarine
    \item 36.- swerve
    \item 37.- verse
    \item 38.- whack
    \item 39.- whop
\end{itemize}
\end{multicols}

     \begin{table}[h]
         \centering
         {
         \scriptsize
         \begin{tabular}{| c | c | c | c || c |}
             \hline
 ~n~ & degree & closeness & betweenness & random-walk \\ \hline
 1 & \{29\} & \{9\} & \{9\} & \{29\} \\ \hline
 2 & \{26, 29\} & \{26, 29\} & \{9, 29\} & \{26, 29\} \\ \hline
 3 & \{10, 12, 26\} & \{10, 12, 26\} & \{9, 26, 29\} & \{10, 12, 26\} \\ \hline
 4 & \{10, 12, 26, 32\} , ... (4) & \{10, 12, 24, 26\} , ... (4) & \{10, 12, 26, 29\} & \{10, 12, 26, 29\} \\ \hline
         \end{tabular}
         }
         \smallskip
         \caption{
             Solutions for group centrality on Roget's Thesaurus sample 2.
         }
         \label{table:thesaurus2}
     \end{table}

%\newpage
%Sample 3 from Roget's Thesaurus network:

\begin{figure}[!ht]
  \centering
    \includegraphics[width=0.7\textwidth]{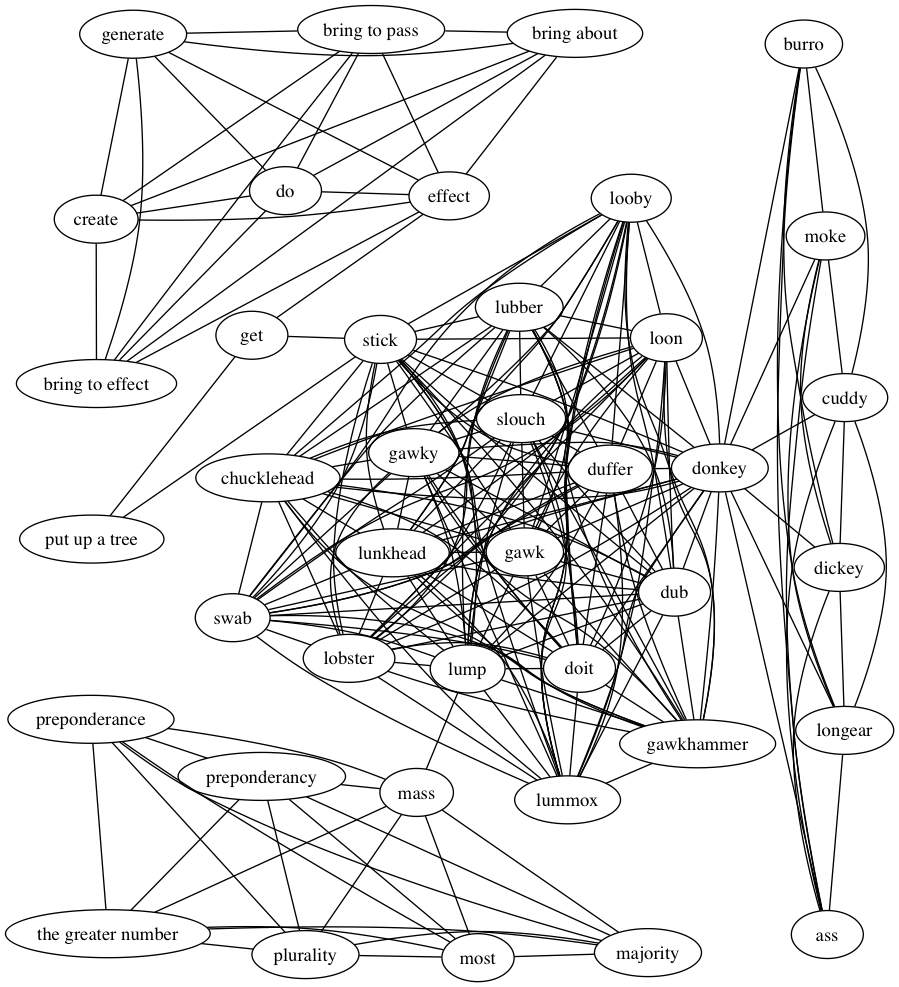}
  \caption{Sample 3 from Roget's Thesaurus network.}
\end{figure}

\newpage
Sample 3 from Roget's Thesaurus network labels:

\begin{multicols}{2}
\begin{itemize}
    \item 0.- ass
    \item 1.- bring about
    \item 2.- bring to effect
    \item 3.- bring to pass
    \item 4.- burro
    \item 5.- chucklehead
    \item 6.- create
    \item 7.- cuddy
    \item 8.- dickey
    \item 9.- do
    \item 10.- doit
    \item 11.- donkey
    \item 12.- dub
    \item 13.- duffer
    \item 14.- effect
    \item 15.- gawk
    \item 16.- gawkhammer
    \item 17.- gawky
    \item 18.- generate
    \item 19.- get
    \item 20.- lobster
    \item 21.- longear
    \item 22.- looby
    \item 23.- loon
    \item 24.- lubber
    \item 25.- lummox
    \item 26.- lump
    \item 27.- lunkhead
    \item 28.- majority
    \item 29.- mass
    \item 30.- moke
    \item 31.- most
    \item 32.- plurality
    \item 33.- preponderance
    \item 34.- preponderancy
    \item 35.- put up a tree
    \item 36.- slouch
    \item 37.- stick
    \item 38.- swab
    \item 39.- the greater number
\end{itemize}
\end{multicols}

    \begin{table}[h!]
        \centering
        {
        \scriptsize
        \begin{tabular}{| c | c | c | c || c |}
            \hline
~n~ & degree & closeness & betweenness & random-walk \\ \hline
1 & \{11\} & \{37\} & \{37\} & \{37\} \\ \hline
2 & \{11, 14\} & \{14, 26\} & \{26, 37\} & \{14, 26\} \\ \hline
3 & \{1, 14, 29\} , ... (6) & \{11, 14, 29\} , ... (6) & \{11, 26, 37\} & \{11, 14, 29\} \\ \hline
4 & \{11, 14, 29, 37\} , ... (440) & \{11, 14, 29, 37\} , ... (440) & \{11, 14, 26, 37\} & \{11, 14, 29, 37\} \\ \hline
        \end{tabular}
        }
        \smallskip
        \caption{
            Solutions for group centrality on Roget's Thesaurus sample 3.
            %this is thesaurus 3
        }
        \label{table:thesaurus3}
    \end{table}


\begin{thebibliography}{10}

\bibitem{DBpedia}
\text{DBpedia}.
\newblock \url{http://www.dbpedia.org}.

\bibitem{HumanMemory}
J. Abbott, J. Austerweil, T. Griffiths.
\newblock Human memory search as a random walk in a semantic network.
\newblock
%In P.~Bartlett et al.
% F.C.N. Pereira, C.J.C. Burges, L.~Bottou, and K.Q.
%  Weinberger, editors,
 {\em Adv. Neur. Inf. Proc. Syst.
  25}, 3050-58, 2012.

\bibitem{Cheng2011}
G. Cheng, Th. Tran, Y. Qu.
\newblock RELIN: Relatedness and Informativeness-based Centrality
for Entity Summarization. ISWC 2011, Springer 2011.

\bibitem{ding2005finding}
L. Ding, R. Pan, T. Finin, A. Joshi, Y. Peng, P. Kolari.
\newblock Finding and ranking knowledge on the semantic web.
\newblock In {\em ISWC 2005}, 156--170. Springer, 2005.

\bibitem{everett}
M. G. Everett, S. P. Borgatti.
\newblock The centrality of groups and classes.
\newblock {\em Journal of Mathematical Sociology}. 23(3): 181-201.

% \bibitem{FastDiscovery}
% C. Faloutsos, K.~S. McCurley, A. Tomkins.
% \newblock Fast discovery of connection subgraphs.
% \newblock {\em 10th ACM SIGKDD international conference
%   on Knowledge discovery and data mining}, KDD '04, 118--127, 2004.


\bibitem{Freeman}
L. C. Freeman.
\newblock Centrality in Social Networks. Conceptual Clarification.
\newblock {\em Social Networks}, 1 (1978/79) 215-239.

% \bibitem{GOH}
% K-I Goh, E~Oh, B~Kahng, and D~Kim.
% \newblock Betweenness centrality correlation in social networks.
% \newblock {\em Physical Review E}, 67(1):017101, 2003.

\bibitem{RankNodesRDF}
A. Graves, S. Adali.
\newblock A method to rank nodes in an rdf graph.
\newblock 7th. ISWC, 2008.

% \bibitem{path1}
% Ch. Halaschek, B. Aleman-Meza, I. Budak Arpinar, A. P. Sheth
% \newblock Discovering and Ranking Semantic Associations over a Large
% RDF Metabase.  30th VLDB, 2004


% \bibitem{MostReliable}
% P. Hintsanen.
% \newblock The most reliable subgraph problem.
% \newblock {\em 11th European Conf. on Principles and
%   Practice of Know. Disc. in Databases}, PKDD 2007, pages 471--478.

\bibitem{LexicalSemantic}
T. Hughes, D. Ramage.
\newblock Lexical semantic relatedness with random graph walks.
{\em  2007 J. Conf. on Emp. Meth. in NLP and Comp. Nat. Lang. Learning, 581-589.}

\bibitem{SamplingLargeGraphs}
Jure Leskovec and Christos Faloutsos.
\newblock Sampling from large graphs.
{\em Proc. 12th ACM KDD '06, pages 631--636.}


% \bibitem{path2}
% K. J. Kochut, M. Janik
% \newblock SPARQLeR: Extended Sparql for Semantic Association Discovery
% \newblock Proc. ESWC 2007.

\bibitem{Lovasz}
L. Lov\'asz
\newblock Random Walks on Graphs; A Survey.
Bolyai Soc., Math. Studies 2, 1993, pp. 1-46.

% \bibitem{Nakao}
% K.~Nakao.
% \newblock Distribution of measures of centrality: Enumerated distributions of
%   freeman's graph centrality measures.
% \newblock {\em Conections}, 13(3):10--22, 1990.

\bibitem{Newman}
Mark~EJ Newman.
\newblock The structure and function of complex networks.
\newblock {\em SIAM review}, 45(2):167--256, 2003.

% \bibitem{DiscoveringMultiRelational}
% C. Ramakrishnan, W.~H. Milnor, M. Perry, A.~P. Sheth.
% \newblock Discovering informative connection subgraphs in multi-relational
%   graphs.
% \newblock {\em SIGKDD Expl. Newsl.}, 7(2):56--63, 2005.

\bibitem{supekar2004characterizing}
K. Supekar, Ch. Patel,  Y. Lee.
\newblock Characterizing quality of knowledge on semantic web.
\newblock {\em AAAI Florida AI Research Symp.
  (FLAIRS-2004)}, 17--19, 2004.

\bibitem{WikiWalk}
E. Yeh, D. Ramage, Ch.~D. Manning, E. Agirre, A. Soroa.
\newblock Wikiwalk: random walks on wikipedia for semantic relatedness.
\newblock {\em 2009 Workshop on Graph-based Meth. for
  NLP}, 41--49, Stroudsburg, PA, USA, 2009.


\bibitem{zhang2007ontology}
X. Zhang, G. Cheng, Y. Qu.
\newblock Ontology summarization based on rdf sentence graph.
\newblock {\em 16th Intl. Conf. on World Wide Web}, 707--716. ACM, 2007.

% \bibitem{NetworkSimplification}
% Fang Zhou, S.~Malher, and H.~Toivonen.
% \newblock Network simplification with minimal loss of connectivity.
%  {\em 10th. IEEE ICDM, 2010, pages 659--668, 2010.}

\bibitem{proximity}
R. Schvaneveldt,  F. Durso, and D. Dearholt.
\newblock Network structures in proximity data.
 {\em The psychology of learning and motivation, 1989, pages 249--284.}

\bibitem{perception}
T. De Smedt, F. De Bleser, V. Van Asch, L. Nijs, W. Daelemans.
\newblock Gravital: natural language processing for computer graphics.
{\em Creativity and the Agile Mind: A Multi-Disciplinary Study of a Multi-Faceted Phenomenon, 2013, pages 81}

\bibitem{zeta}
Howard Levinson
\newblock An Eigenvalue Representation for Random Walk Hitting Times and its Application to the Rook Graph.
 {\em }

\bibitem{goodcita}
S. Wasserman, K. Faust.
\newblock Social Netowrk Analysis: Methods and Applications (Structural Analysis in the Social Sciences).
{\em Cambridge Univ. Press, 1994.}
\end{thebibliography}
\end{document}